\documentclass[final,leqno]{siamltex}
\usepackage{algorithm,algorithmic}
\usepackage{amsmath,graphicx}
\usepackage{epstopdf}
\usepackage{amssymb}
\usepackage{multirow}
\usepackage{epstopdf}
\usepackage{graphicx}
\usepackage{xcolor}
\usepackage{booktabs}

\graphicspath{{../figures/} }

\def\x{{\mathbf x}}
\def\y{{\mathbf y}}
\def\m{{\mathbf m}}

\def\Expect{\mathbb{E}}
\def\R{\mathbb{R}}
\def\C{{\mathbb{C}}}

\def\N{\mathcal{N}}

\def\M{{\cal{M}}}
\def\B{{\cal B}}
\def\Id{\mbox{Id}}
\def\Prob{\mathbb{P}}
\newcommand{\defeq}{\mathrel{\mathop:}=} 

\DeclareMathOperator*{\argmax}{argmax}
\DeclareMathOperator*{\argmin}{argmin}

\DeclareMathOperator*{\LUT}{LUT}
\DeclareMathOperator*{\Real}{real}
\DeclareMathOperator*{\TE}{TE}
\DeclareMathOperator*{\TR}{TR}
\DeclareMathOperator*{\Tone}{T1}
\DeclareMathOperator*{\Ttwo}{T2}

\newtheorem{Proposition}{Proposition}
\newtheorem{Theorem}{Theorem}
\newtheorem{Lemma}{Lemma}
\newtheorem{Definition}{Definition}

\newtheorem{Remark}{Remark}

\begin{document}
%
\title{A Compressed Sensing Framework for Magnetic Resonance Fingerprinting}
%
%
\author{Mike Davies, Gilles Puy, Pierre Vandergheynst and Yves Wiaux}
\maketitle

\begin{abstract}
Inspired by the recently proposed Magnetic Resonance Fingerprinting (MRF) technique, we develop a principled compressed sensing framework for quantitative MRI. The three key components are: a random pulse excitation sequence following the MRF technique; a random EPI subsampling strategy and an iterative projection algorithm that imposes consistency with the Bloch equations. We show that theoretically, as long as the excitation sequence possesses an appropriate form of persistent excitation, we are able to accurately recover the proton density, $\Tone$, $\Ttwo$  and off-resonance maps simultaneously from a limited number of samples. These results are further supported through extensive simulations using a brain phantom.
\end{abstract}


\section{Introduction}

Inspired by the recently proposed procedure of Magnetic
Resonance Fingerprinting (MRF), which gives a new technique for
quantitative MRI, we investigate this idea from a compressed
sensing perspective. While MRF itself, was inspired by the recent growth of compressed sensing (CS) techniques
in MRI \cite{Ma-MRF2013}, the exact link to CS was not made explicit, and the
paper does not consider a full CS formulation. Indeed the role of sparsity, random excitation and sampling are not clarified. 
The goal of this current paper is to make the links with CS explicit, shed light on the appropriate acquisition and reconstruction procedures and hence to develop a full compressed sensing strategy for quantitative MRI.


In particular, we identify separate roles for the pulse excitation and the subsampling of $k$-space. We
identify the Bloch response manifold as the appropriate low dimensional signal model on which
the CS acquisition is performed, and interpret
the ``model-based'' dictionary of \cite{Ma-MRF2013} as a natural discretization of this response manifold
We also discuss what is necessary in order to have an appropriate CS-type
acquisition scheme.  

Having identified the underlying signal model we next turn to the reconstruction process. In \cite{Ma-MRF2013} this
was performed through pattern matching using a matched filter based on the model-based 
dictionary. However, this does not offer the opportunity for exact
reconstruction, even if the signal is hypothesised to 
be 1-sparse in this dictionary due to the undersampling of $k$-space. This suggests that we should look to a model based CS framework that directly supports such manifold models \cite{Baraniuk-Wakin2009}. Recent algorithmic work in this direction has been presented by 
Iwen and Maggioni \cite{Iwen-Maggioni-2012}, however, their approach is not practical in the present context as the computational cost of their scheme grows exponentially with the dimension of the manifold. Instead, we
leverage recent results from \cite{TB-2011} and develop a recovery algorithm based on the Projected Landweber Algorithm (PLA). This method also has the appealing interpretation of an iterated refinement of the original MRF scheme.

The remainder of the paper is set out as follows. We begin by giving a brief overview of MRI acquisition. Then we discuss the challenges of quantitative imaging in MRI and review the recently proposed MRF scheme \cite{Ma-MRF2013}. We next develop the detailed mathematical model associated with MRF acquisition which leads us to the voxel-wise Bloch response manifold model observed through a sequence of partially sampled $k$-space measurements. In  
\S\ref{sec: CS-MRF}, using the MRF acquisition model, we set out a framework for a compressed sensing solution to the quantitative MRI problem followed by a simple extension that provides a degree of spatial regularization. 

In the simulation section we demonstrate the efficacy of our methods on an anatomical brain phantom \cite{Collins-1998}, available at the BrainWeb repository \cite{Brainweb}. Our results show that our CS method offers substantial gains in reconstruction accuracy over the original MRF matched filter scheme \cite{Ma-MRF2013}. We also demonstrate the efficiency of the proposed algorithm in terms of speed of convergence and the empirical trade-off between undersampling in $k$-space and excitation sequence length.

Finally, we summarize what we have learnt by placing the MRF procedure within a CS framework and highlight a number of open questions and research challenges.

%

\section{Magnetic Resonance Imaging Principles}

MRI, with its ability to image soft tissue, provides a very powerful imaging tool for medicine. The basic principles of MRI lie in the interaction of proton spins with applied magnetic fields. While a full review of these principles is beyond the scope of this paper, following \cite{Wright1997}, we now introduce the basics required in order to understand the motivation for the proposed acquisition scheme and the subsequent mathematical models. For a more detailed treatment of MRI from a signal processing perspective we refer the reader to one of the excellent reviews on the subject, such as \cite{Wright1997,Fessler2010}.

\subsection{Bloch Equations}
The main source of the measured signal in MRI comes from the magnetic moments of the proton spins. In a single volume element (voxel) the net magnetization $\m = (m^x,m^y,m^z)^T$ is the vector sum of all the individual dipole moments within the voxel. If there is no magnetic field then at equilibrium the net magnetization is zero. 

If a static magnetic field, ${\bf B}_0$ (usually considered to lie in the $[0,0,1]^T$ direction), is then applied the spins align with this field and the net magnetization at equilibrium, $\m_{\mbox{eq}}$, is proportional to the proton density $\rho$ within the volume. However, equilibrium is not achieved immediately after the field is applied, but is controlled by the longitudinal relaxation time, $\Tone$, such that the net magnetization at time $t$ is given by: $m^z(t) = m_{\mbox{eq}} (1- \exp(-t/\Tone))$.

If there is magnetization in the plane orthogonal to ${\bf B}_0$ then the magnetization, $\{m^x,m^y\}$, precesses about the $z$ axis at a frequency called the Lamor frequency, $\omega_L = \gamma |{\bf B}_0|$ (approximately $42.6$MHz per Tesla), where the quantity $\gamma$ is called the gyromagnetic ratio. This in turn emits an electromagnetic signal which is the signal that is measured. As the individual dipoles dephase the net transverse magnetization decays exponentially at a rate $\Ttwo$, called the transverse relaxation time.

In MRI the magnetic field is composed of a static magnetic field and a dynamic component which is manipulated through a radio frequency (RF) coil aligned with the $x$ direction. When a transverse magnetic field is applied via an RF pulse the proton dipoles rotate about the applied magnetic field. The overall macroscopic dynamics of the net magnetization can be summarized by a set of linear differential equations called the Bloch equations \cite{Wright1997}:
\begin{equation}
\frac{\partial \m(t)}{\partial t} = \m(t) \times \gamma {\bf B}(t) -\left(\begin{array}{c}
m^x(t)/\Ttwo\\
m^y(t)/\Ttwo\\
(m^z-m_{\mbox{eq}})/\Tone
\end{array}
\right)
\end{equation}
The response at a given readout time ($\TE$) from an initial RF pulse can be determined by integrating these equations over time. When a specific sequence of pulses is applied (assuming pulse length $\ll \Tone, \Ttwo$) then the dynamics of the magnetization from pulse to pulse or readout to readout can be described simply by a three dimensional discrete time linear dynamical system \cite{Jaynes1955}.

\subsection{Spatial Encoding and Image Formation}

In  order to produce an image it is necessary to spatially encode the magnetization in the received signal. This is done through the application of various magnetic gradients.
First, a slice can be selected through the application of a magnetic gradient along the $z$ direction, while appropriately restricting the frequency band of the excitation pulses. The gradient changes the Larmor frequency as a function of $z$, and only those positions that are excited by the pulses generate a magnetization in the transverse plane.

In order to encode the transverse magnetization spatially at the acquisition time (called the echo time (TE)) the magnetic field can be modified further to have gradients $G_x$ and $G_y$ in the $x$ and $y$ directions. For example, if a linear gradient is applied along the $x$ direction so that $B^z = (B_0+G_x x)$, then the spatial variation of the transverse magnetization is encoded in the Larmor frequency and hence in the frequency of the received signal (it is assumed that the duration of the signal read out time is sufficiently short such that the magnetization can be treated as stationary). The received signal therefore corresponds to a line in the spatial Fourier transform, known as $k$-space, of the transverse magnetization. By careful selection of $G_x$ and $G_y$ it is possible to sample different lines of $k$-space until it is adequately sampled. The most popular technique is to take measurements, which we denote $\y$, that sample $k$-space on a Cartesian grid, so that the image can be formed by the application of the inverse 2D Discrete Fourier transform (DFT), $F$. Thus we can generate a discrete image $\x$ (represented here in vector form) by using $\x=F^H \y$, where $H$ denotes the conjugate transpose. For simplicity, unless stated otherwise, we will work with this discrete representation and assume that samples in $k$-space have been taken on the Cartesian grid.

\subsection{Rapid Imaging}

A key challenge in MRI is acquiring the signals in a reasonably short time. Long scan times are costly, unpopular with patients and can introduce additional complications such as motion artefacts. However, the set up described so far for MRI requires the application of repeated excitation pulses and gradients to acquire the multiple lines of $k$-space. Furthermore, after each acquisition sufficient time must be left in order for the magnetization to achieve equilibrium once again.

One way to accelerate the imaging is to acquire more samples from $k$-space per acquisition. By varying the transverse gradients $G_x$ and $G_y$ as a function of time it is possible to generate more sophisticated sampling patterns. For example, in echo-planar imaging (EPI) \cite{McKinnon: multi-shot EPI} multiple lines of $k$-space are acquired at each pulse. Another strategy is to generate spiral trajectories  in $k$-space. However, in both cases as the readout time gets longer artefacts are introduced by variation in the transverse magnetization over the read out time. Furthermore, in the case of spiral and other non-Cartesian trajectories there is the added complication of requiring more complicated image formation algorithms, such as gridding techniques \cite{Bydder2007}, that attempt to approximate the pseudo-inverse of the non-uniform Fourier transform \cite{Fessler2003}. 
 
A second approach to rapid imaging is to take fewer samples.  Since the emergence of compressed sensing in MRI \cite{Lustig2008}, the idea of subsampling $k$-space has become very popular.
Compressed Sensing exploits the fact that the image being acquired can be approximated by a low dimensional model, e.g. sparse in the spatial or wavelet domain. Then, under certain circumstances, the image can be recovered from a subsampled $k$-space using an appropriate iterative reconstruction algorithm.

Parallel imaging techniques can also be used in conjunction with the above strategies to provide further acceleration. However, these are outside the scope of the current work.

\subsection{Quantitative MRI}

Rather than simply forming an image that measures the transverse magnetization response from a single excitation, the aim of quantitative imaging is to provide additional physiological information by estimating the spatial variation of one or more of the  physical parameters that control the Bloch equations, specifically: $\Tone$, $\Ttwo$ and proton density. These can help in the discrimination of different tissue types and provide useful information in numerous application areas, such as diffusion and perfusion imaging. 

The standard approach to parameter estimation is to acquire a large sequence of images in such a way that for each voxel the sequence of values is dependent on either $\Tone$ and/or $\Ttwo$ as well as certain nuisance parameters. For example, the most common techniques acquire a sequence of images at different echo times from an initial excitation pulse. For $\Tone$ estimation, this is typically an inversion recovery pulse (full $180^{\circ}$ rotation of the magnetic field) and for $\Ttwo$ it is a spin-echo pulse ($90^{\circ}$ rotation). The image sequences encode the exponential relaxation and the parameter of interest can be estimated by fitting an exponential curve to each voxel sequence. Another approach \cite{Deoni-DESPOT-2003} uses a set of well tailored steady state sequences, such that each voxel sequence encodes the relevant parameter values. Such techniques require the acquisition of multiple lines for multiple images, and it is very challenging to achieve within a reasonable time and with an acceptable signal-to-noise ratio (SNR) and resolution.
 
Recently there have been a number of papers attempting to address this problem taking a compressed sensing approach \cite{Block2009,Doneva2010,Zhao2013,Tran-Gia2013,Huang2013}. All these techniques accelerate the parameter acquisition using an exponential fitting model combined with only partially sampling $k$-space for each image. Model-based optimization algorithms \cite{Fessler2010} are then used to retrieve the parameter values. However, while such approaches take their inspiration from compressed sensing and exploit sparse signal models, these techniques mainly focus on the development of novel reconstruction algorithms, and do not tackle the fundamental issue of how to design the acquisition in such a way as to meet the compressed sensing sampling criteria. 

In contrast to this previous body of work, here we will set out a principled compressed sensing approach to the simultaneous determination of all parameters of interest. That is, we will develop an acquisition framework that can be shown to satisfy the compressed sensing criteria thereby enabling us to develop a model based parameter estimation algorithm with exact recovery guarantees. The basis of our acquisition scheme is the recently proposed `magnetic resonance fingerprinting' technique \cite{Ma-MRF2013} which we describe next.

\section{Magnetic Resonance Fingerprinting}

In the recent paper \cite{Ma-MRF2013} a new type of MRI acquisition scheme is presented that enables the quantification of multiple tissue properties simultaneously through a single acquisition process. The procedure is composed of 4 key ingredients: 
\begin{enumerate}
\item 
The material magnetization is excited through a sequence of random RF pulses. There is no need to wait for the signal to return to equilibrium between pulses or for the response to reach a steady state condition as in other techniques. 
\item
After each pulse the response is recorded through measurements in $k$-space. Due to the time constraints only a proportion of $k$-space can be sampled between each pulse. In \cite{Ma-MRF2013} this is achieved through Variable Density Spiral (VDS) sampling. 
\item
A sequence of magnetization response images is formed using gridding to approximate the least square solution. These images suffer from significant aliasing due to the high level of undersampling.
\item
Parameter maps (proton density, $\rho$, $\Tone$, $\Ttwo$ and off-resonance,\footnote{The off-resonance frequency is a additional parameter that can be incorporated into the Bloch equations and measures local field inhomogeneity and chemical shift effects \cite{Ganter}.} $\delta f$) are formed through a pattern matching algorithm that matches the alias-distorted magnetization response sequences per voxel to the response predicted from the Bloch equations.
\end{enumerate}
 
Below we develop the relevant mathematical models for the MRF acquisition system that will allow us to develop a full CS strategy for quantitative MRI.

\subsection{Pulse excitation and the Bloch response manifold}
The MRF process is based upon an Inversion Recovery Steady State Free
Precession (IR-SSFP) pulse sequence.\footnote{As the excitation pulses in MRF are random the term steady state is now somewhat of a misnomer and we should possibly call these Inversion Recovery Randomly Excited Free Precession.} The dynamics of the magnetization for each voxel, assuming a single chemical composition,
are described by the response of the Bloch
equations when 'driven' by the excitation parameters.
 
Let $i = 1, \ldots, N$ index the voxels of the imaged slice. The MRF excitation
generates a magnetization response that can be observed (or at least partially observed) at each excitation pulse. The magnetization at a
given voxel at the $l$th echo time is then a function of the excitation parameters of the $l$th excitation pulse, the magnetization at the $(l-1)$th echo time, the overall magnetic field and the unknown parameters associated with the given voxel. The overall dynamics can be described by a parametrically excited linear system and are summarized in appendix \ref{app: bloch dynamics}. 

The magnetization dynamics at voxel $i$ are parameterized by the voxel's parameter set ${\bf
  \theta }_i = \{\Tone_i,\Ttwo_i,\delta f_i \} \in \M$, where $\M \subset \R^3$ denotes the set of feasible values for $\theta_i$, and the voxel's proton density, $\rho_i$.
 The magnetization response dynamics are also characterized by the excitation parameters of the $l$th pulse, namely the
 flip angle, $\alpha_l$, and the repetition time, $\TR_l$. 
   
 Now and subsequently we will denote the magnetization image sequence by the matrix $X$, with $X_{i,l}$ denoting the magnetization for voxel $i$ at the $l$th read out time. Note we are representing the response image at a the $l$th readout by a column vector which we denote as: $X_{:,l}$, using a Matlab style notation for indexing. Similarly, we will denote the magnetization response sequence for a given voxel $i$ as $X_{i,:}$.
 
Given the initial magnetic field, the initial magnetization of any voxel is known up to the
unknown scaling by its proton density $\rho_i$. Thus the magnetization response at any voxel can be written as a parametric nonlinear mapping
 from $\{\rho_i, {\theta}_i\}$ to the sequence, $X_{i,:}$:
\begin{equation}\label{eq: bloch map}
X_{i,:} = \rho_i B({\bf \theta}_i; \alpha, \TR) \in \C^{1\times L}.
\end{equation}
Here $ \rho_i \in \R_+$ is the proton density at voxel $i$, $L$ is 
the excitation sequence length and $B$ is a smooth mapping
induced by the Bloch equation dynamics: $B:{ \cal{M}} \rightarrow
\C^{1\times L}$, where its smoothness can be deduced by the smooth dependence of the dynamics \eqref{eq: magnetization response} and \eqref{eq: TE response} with respect to $\theta_i$.

In order to be able to retrieve the Bloch parameters ${\bf
  \theta}_i$ and proton density from $X_{i,:}$ it is necessary that the excitation
sequence is ``sufficiently rich'' such that the voxel's magnetization response \eqref{eq: bloch map} can be distinguished from a response with different parameters.
Mathematically this means that there is an embedding of $\R_+ \times \M$ into $\C^L$.\footnote{Strictly speaking we can only consider this to be an embedding for $\rho_i >0$ otherwise $\theta_i$ is not observable.} We will call $\B = B(\M; \alpha, \TR) \subset \C^L$ the Bloch response manifold and denote its cone by $\R_+\B$. 

\begin{Remark} Note that this component of the MRF procedure is not compressive, as the mapping \eqref{eq: bloch map} will typically need to map to a higher dimension than $\dim (\R_+\M)$ in order to induce an embedding. The primary role of the excitation sequence is therefore to ensure identifiability and this can typically be achieved through random excitation as is commonly used in system identification. We will see, however, that the excitation sequence will also need to induce a sufficiently persistent excitation in order for it to be observed in a compressive manner.
\end{Remark}

\begin{Remark} The aim of a good excitation sequence should be to minimize
  the time taken to acquire the necessary data rather than minimizing the total
  number of samples. To this end, the total acquisition time for the
  sequences, $\sum_l \TR_l$ is the relevant cost. Here, while more samples may be taken in MRF in comparison with other quantitative techniques the benefit comes in not having to wait for the magnetization to relax to its equilibrium state between samples. 
\end{Remark}

\begin{Remark} \label{rem: proton density}
While it is clear that the proton density, $\rho_i$, will necessarily be real valued and non-negative, it is common practice in MRI to allow this quantity to absorb additional phase terms due to, for example, coil sensitivity or timing errors. Therefore $\rho_i$ is often allowed to take a complex value. In this work we will retain the idealized model, treating it as non-negative real, however, we note that the subsequent theory presented here can typically be easily modified to work with $\rho_i \in \C$ instead of $\rho_i \in \R_+$, albeit with an increase in the dimensionality of the unknown parameter set. We will highlight specific differences along the way.
\end{Remark}

\subsection{MRF imaging and $k$-space sampling}

So far we have considered the signal model for a single voxel. For a complete spatial image, assuming a discretization into $N$ voxels
and treating each voxel as independent we have ${\bf \theta} \in \M^N$ and $\rho \in \R_+^N$. Similary $X \in
\C^{N\times L}$. We can therefore define the full response mapping, $X = f(\rho,\theta)$, $f: \R_+^N \times \M^N \rightarrow (\R_+\B)^N \subset \C^{N\times L}$, as:
\begin{equation}
\label{eq: full response map}
X = f(\rho,\theta) = [\rho_1 B({\bf \theta}_1; \alpha, \TR), \ldots, \rho_N B({\bf \theta}_N; \alpha, \TR)]^T.
\end{equation}

Unfortunately, it is impractical to observe the full spatial magnetization (via $k$-space) at each repetition time within a sufficiently small time for the magnetization to remain approximately constant. It is therefore necessary to resort to some form of undersampling. Let us denote the observed sequence of $k$-space samples as $Y \in \C^{M \times L}$, such that the samples taken at the $l$th read out,  $Y_{:,l} \in \C^M$ are given by:
\begin{equation}\label{eq: observation map}
Y_{:,l} = P(l) F X_{:,l}
\end{equation}
where $F$ again denotes the 2D discrete Fourier transform and $P(l)$ is the
projection onto a subset of coefficients measured at the $l$th read out (although the original MRF scheme used a sequence of spiral trajectories, for simplicity we will 
assume that the Fourier samples are only taken from a Cartesian grid). 
We can finally define the full linear observation map from the spatial magnetization sequence to the observation sequence as $Y=h(X)$ where $h$ is given by:
\begin{equation}
\label{eq: full observation map}
Y = h(X) = [P(1)F X_{:,1}, \ldots, P(N)FX_{:,N}].
\end{equation}
Together \eqref{eq: full response map} and \eqref{eq: full observation map} define the full MRF acquisition model from the parameter maps $\Tone, \Ttwo, \delta f$ and $\rho$ to the observed data $Y$.

\subsection{MRF matched filter reconstruction}
\label{sec: MRF reconstruction}
In \cite{Ma-MRF2013} the image sequence is first reconstructed using the regridding method \cite{Bydder2007} which approximates the least squares estimate for $X_{i,:}$ given $Y_{i,:}$:
\begin{equation}\label{eq: MRF image BP}
\hat{X}_{:,l} = F^H P(t)^T Y_{:,l} 
\end{equation}
or equivalently $\hat{X} = h^H (Y)$.  Due to the high level of undersampling, each reconstructed image contains significant aliasing. However, it is argued in \cite{Ma-MRF2013} that accurate estimates of the parameter maps can still be obtained by matching each voxel sequence to a predicted Bloch response sequence using a set of matched filters. This essentially averages the aliasing across the sequence, treating the aliasing as noise. While the technique provides impressive results, it ignores the main tenet of compressed sensing - that aliasing is interference and under the right circumstances can be completely removed (we explore this idea in detail in \S\ref{sec: CS-MRF}).

Mathematically, it will be convenient to view the matched filter solution as the projection of the voxel sequence onto a discretization of the Bloch response manifold as follows.

\subsubsection{Sampling the Bloch response manifold}
\label{sec: Bloch projection}


Suppose that we wished to approximate the projection of the sequence $X_{i,:}$ onto the cone of the Bloch response manifold.
One way to do this is to first take a discrete set of samples of the parameter space,  $\M$,  $\theta_i^{(k)} = \{\Tone_i^{(k)},\Ttwo_i^{(k)},\delta f_i^{(k)} \}$, ${k = 1, \ldots, P}$ and 
construct a `dictionary' of magnetization responses, $D = \{D_k\}$, $D_k =  B(\theta_i^{(k)};
\alpha, \TR)$, $k = 1, \ldots , P$. The density of such samples controls the accuracy of the final approximation of the projection operator.

We can similarly construct a look-up table (LUT) to provide an inverse for $B(\theta_i;
\alpha, \TR)$ on the discrete samples such that $\theta_i^{(k)} = \LUT_B (k)$. 

The projection onto the cone of the discretized response
manifold, $D$, can then be calculated using:
\begin{equation}\label{eq: bloch MF}
\hat{k}_i = \argmax_k \frac{\Real \langle D_k,X_{i,:}\rangle}{\|D_k\|_2}
\end{equation}
to select the closest sample $D_{\hat{k}_i}$ and 
\begin{equation}\label{eq: bloch MF proton density}
\hat{\rho}_i = \max\{\Real \langle D_{\hat{k}_i},X_{i,:}\rangle/\|D_{\hat{k}_i}\|_2^2,0\}
\end{equation}
for the proton density, where the real and max operations are necessary to select only positive correlations since negative $\rho_i$ are not admissible.  

If we allow $\rho_i$ to be complex valued (see Remark~\ref{rem: proton density}) then the projection equations become:
\begin{equation}\label{eq: bloch complex MF}
\hat{k}_i = \argmax_k \frac{| \langle D_k,X_{i,:}\rangle|}{\|D_k\|_2}
\end{equation}
and
\begin{equation}\label{eq: bloch MF complex proton density}
\hat{\rho}_i = 
\langle D_{\hat{k}_i},X_{i,:}\rangle/\|D_{\hat{k}_i}\|_2^2
\end{equation}

Equations \eqref{eq: bloch complex MF} and \eqref{eq: bloch MF complex proton density} are precisely the matched filter equations used in \cite{Ma-MRF2013}, applied to the distorted voxel sequences. We therefore see that one interpretation of matched filtering with the MRF dictionary model is to provide an approximate projection onto the cone of the Bloch response manifold for each voxel sequence.

A summary of the full MRF parameter map recovery algorithm (with a real valued proton density model) is given in Algorithm~\ref{alg: MRF}.

\begin{algorithm}[t]
 \caption{MRF reconstruction}\label{alg: MRF}
\begin{algorithmic}[0]
 \STATE \textbf{Given:} $Y$
\STATE Reconstruct $X$: 
\STATE $\hat{X} = h^H (Y)$
\STATE MF parameter estimation:
   \FOR{ $i=1:N$}
    \STATE   $\hat{k}_i = \argmax_k \Real \langle D_k,\hat{X}_{i,:}\rangle/\|D_{k}\|_2$
    \STATE   $\hat{\theta}_i = \LUT_{\B}(\hat{k}_i)$
    \STATE   $\hat{\rho}_i = \max\{0,\Real \langle D_{\hat{k}_i},\hat{X}_{i,:}\rangle/\|D_{\hat{k}_i}\|_2^2\}$
   \ENDFOR
\STATE \textbf{Return:} $\hat{\theta}, \hat{\rho}$
\end{algorithmic}
\end{algorithm}

\subsubsection*{Computational cost and accuracy}
Given that the discretized MRF dictionary can be very large ($\approx 500,000$ samples in \cite{Ma-MRF2013}), it is useful to consider the computational complexity of the above calculations as a function of parameter accuracy as this is the major computational bottleneck that we will encounter.

The accuracy with which we can estimate the parameters for a given voxel will depend on the accuracy of the approximate projection operator and the Lipschitz constants of the inverse mapping, $LUT_{\B}$.  We can achieve an approximate projection by generating an $\epsilon$-cover of $\B$ with $D_k$. As the dimension of $\B$ is $3$, this requires choosing $P \sim C\epsilon^{-3}$ atoms in our dictionary. Furthermore, as the projection operation described in \eqref{eq: bloch MF} takes the form of a nearest neighbour search, we can use fast nearest neighbour search strategies, such as the cover tree method \cite{Beygelzimer2006}, to quickly solve \eqref{eq: bloch
  MF} in $\mathcal{O}(L\ln (1/\epsilon))$ computations per voxel, instead of
the $\mathcal{O}(L\epsilon^{-3})$ necessary for exhaustive search. This effectively makes the speed of each application of $D$  on a par with that of a traditional fast transform. Similarly, the approximate inverse using $\LUT_\B$ can also be computed in $\mathcal{O}( \ln (1/\epsilon))$.

We could also consider enhancing such an estimate by exploiting the smoothness
of the response manifold, either by using local linear approximations of the manifold
\cite{Iwen-Maggioni-2012} or  by further locally optimizing the
projection numerically around the selected parameter set, once we are assured global convergence.
Such an enhancement could allow either for increased accuracy or reduced computation through the 
use of fewer parameter samples, however, we do not pursue these ideas further here.

\section{Compressed Quantitative Imaging}\label{sec: CS-MRF}

In order to generate a full compressed sensing framework for MRF we will identify sufficient conditions on the excitation pulse sequences and the $k$-space sampling, along with a suitable reconstruction algorithm, to guarantee recovery of the parameter maps from the observed $k$-space samples. As the dimension of our problem is large, $\dim((\R_+\times \M)^N) = 4N$, we do not consider the manifold reconstruction algorithms in \cite{Iwen-Maggioni-2012} as these scale poorly with the dimension of the manifold. 
Instead, we propose a CS solution based around the iterative projection algorithm of Blumensath \cite{TB-2011} which we will see has computational cost that is linear in the voxel dimension.
Our approach, which we call BLIP (BLoch response recovery via Iterated Projection), has three key ingredients:
a random pulse excitation sequence following the original MRF technique; a random subsampling strategy that can be shown to induce a low distortion embedding of $\R_+^N \times \M^N$ and an efficient iterated projection algorithm \cite{TB-2011} that imposes consistency with the Bloch equations. 
Moreover, the projection operation is the same nearest neighbour search described in section~\ref{sec: Bloch projection}. 

We first describe the iterative projection method and then consider the implications for the appropriate excitation and sampling strategies. 

\subsection{Reconstruction by Iterated Projection}
 
In \cite{TB-2011} a general reconstruction algorithm, the Projected
Landweber Algorithm (PLA) was proposed as an extension of the popular Iterated Hard
Thresholding Algorithm \cite{TB-IHT-2008,TB-IHT-2009}. PLA is applicable to
arbitrary union of subspace models as long as we
have access to a computationally tractable projection operator onto the union of subspace model within
the complete signal space. The algorithm is given by:
\begin{equation}
X^{(n+1)} = {\cal P}_{\cal A}(X^{(n)}+\mu h^H(Y-hX^{(n)}))
\end{equation}
where ${\cal P}_{\cal A}$ is the orthogonal projection onto the signal
model ${\cal A}$ such that 
\begin{equation}
{\cal P}_{\cal A}(X) = \argmin_{\tilde{X}\in {\cal A}} \|X-\tilde{X}\|_F
\end{equation}
and $\mu$ is the step size. 

The current theory for PLA  \cite{TB-2011} states that a sufficient condition for stable recovery of $X$ given $Y$ is that $h$ is a stable embedding
 - a so-called
Restricted Isometry Property (RIP) or bi-Lipshitz embedding - for the signal
model, ${\cal A}$. A mapping, $h$, is said to have the RIP (be a bi-Lipschitz embedding) for the signal model ${\cal A}$ if there exists a sufficiently small constant $\delta >0$
such that:
\begin{equation}\label{eq: bi-lip embedding}
(1-\delta) \| X-\tilde{X}\|_2^2 \leq \frac{N}{M}\|h(X-\tilde{X})\|_2^2 \leq (1+\delta) \| X-\tilde{X}\|_2^2 
\end{equation}
for all pairs $X$ and $\tilde{X}$ in ${\cal A}$. How to achieve such an embedding will be considered later in section~\ref{sec: excitation and sampling}.

The theory \cite{TB-2011} states that it is sufficient that $h$ satisfy the RIP with $\frac{M}{N}(1+\delta) < 1/\mu < \frac{3M}{2N} (1-\delta)$ for the guaranteed recovery.
If $h$ is essentially `optimal', e.g. a random ortho-projector, then we should set the step size $\mu \approx N/M$ since in the large system limit $\delta \rightarrow 0$. 

For our compressed sensing scenario the signal model ${\cal A}$ is the product set $(\R_+ \B)^N$ or, more precisely,  its discrete approximation $(\R_+ D)^N$ and the projection operator ${\cal P}_{\cal A}$ can be realized by separately projecting the individual voxel sequences $X_{i,:}^n$ onto the cone of the Bloch response manifold using the equations \eqref{eq: bloch MF} and \eqref{eq: bloch MF proton
  density}. 
Although $(\R_+ \B)^N$ is not itself a union of subspace model it can easily be extended to $(\R \B)^N$, which forms an uncountably infinite union of lines (1D subspaces). In fact, the theory of \cite{TB-2011} does not require ${\cal A}$ to be a union of subspace \cite{TB-2011} and is directly applicable to ${\cal A} = (\R_+ \B)^N$.  We therefore appear to have all the ingredients for a full compressed sensing recovery algorithm. This is summarized in Algorithm~\ref{alg:CS-MRF}. 
\begin{algorithm}[t]
 \caption{BLoch response recovery via Iterative Projection (BLIP) }\label{alg:CS-MRF}
\begin{algorithmic}[0]
\STATE \textbf{Given:} $Y$
 \STATE \textbf{Initialization:} $X^{(0)} = \mathbf{0}$, $\mu = N/M$
 \STATE Image sequence reconstruction
 \FOR{ $n=1; n:=n+1$ \textbf{until stopping criterion}}
\STATE Gradient step:
   \FOR {$l = 1:L$}
     \STATE $X_{:,l}^{(n+1/2)} = X_{:,l}^{(n)}+\mu F^H P(l)^T(Y_{:,l}-P(l) F
     X_{:,l}^{(n)})$;
   \ENDFOR
\STATE Projection step:
   \FOR{ $i=1:N$}
    \STATE   $\hat{k}_i = \argmax_k \Real \langle D_k,X_{i,:}^{(n+1/2)}\rangle/\|D_{k}\|_2$
    \STATE   $\hat{\rho}_i = \max\{0,\Real \langle D_{\hat{k}_i},X_{i,:}^{(n+1/2)}\rangle/\|D_{\hat{k}_i}\|_2^2\}$
    \STATE   $X_{i,:}^{(n+1)} = \hat{\rho}_i D_{\hat{k}_i}$
   \ENDFOR
 \ENDFOR
 \STATE Parameter map estimation:
   \FOR{ $i=1:N$}
    \STATE   $\hat{\theta}_i = \LUT_{\B}(\hat{k}_i)$
   \ENDFOR
   \STATE \textbf{Return:} $\hat{\theta}, \hat{\rho}$
\end{algorithmic}
\end{algorithm}


\begin{Remark} \label{remark: separating CS and parameter estimation} Note that the above procedure has separated out the parameter map estimation (by inverting the estimated Bloch responses) and the reconstruction of the magnetization image sequence (via the PLA). Indeed, as long as the partial $k$-space sampling provides a bi-Lipschitz embedding for all possible magnetization responses then the CS component of the imaging is well defined even if the Bloch response is not invertible.
\end{Remark}

\subsubsection{Step size selection}
\label{sec: adaptive step size}

Selection of the correct step size is crucial in order to attain good performance from these iterative projection based algorithms \cite{TB-NIHT-2010,TB-2011}. Note that the original parameter estimation in \cite{Ma-MRF2013} can be interpreted as an application of a single iteration of PLA with a step size $\mu = 1$ and iterating PLA with this step size tends to only deliver a modest improvement over the matched filter (single iteration). The matched filter also has the effect of underestimating the magnitude of $X$, and hence also the proton density map, as $h$ tends to shrink vectors uniformly (when it provides a stable embedding).

In contrast, when using the substantially more aggressive step size proposed by the theory we will see that significant improvements are observed in signal recovery and often in a very small number of iterations.

In practice, it is also beneficial to select the step size for PLA adaptively to ensure stability. Following the work on adaptive step size selection for IHT \cite{TB-NIHT-2010} we adopt the following heuristic. We begin each iteration by choosing $\mu = {N/M}$ as is suggested from the CS theory. Then after calculating a new proposed value for $X^{n+1}$ we calculate the quantity:
\begin{equation}
\omega = \kappa \frac{\|X^{n+1}-X^{n}\|_2^2}{\|h(X^{n+1}-X^{n})\|_2^2}
\end{equation}
for some $\kappa < 1$. If $\mu > \omega$ we reject this update, shrink the step size, $\mu \mapsto \mu/2$ and calculate a new proposed value for $X^{n+1}$. As with the Normalized IHT algorithm \cite{TB-NIHT-2010}, this form of line search is sufficient to ensure convergence of the algorithm irrespective of conditions on the measurement operator, and we will use this form of step size selection in all subsequent experiments.

\subsection{Strategies for subsampling $k$-space}
\label{sec: excitation and sampling}

We now consider what properties of the excitation response sequences and the $k$-space sampling pattern will ensure that the sufficient RIP conditions in the PLA theory are satisfied.


First note that, as the signal model treats each voxel as independent, we need to take at least $N \dim(\R_+ \M)$ measurements as this is the dimension of our model. Furthermore, since we only take a small number of measurements at each repetition time, we cannot expect to achieve a stable embedding without imposing further constraints on the excitation response. For example, if the embedding was induced in the first few repetition times and all further responses were non-informative we would not have taken sufficient measurements from the informative portion of the response. Therefore we consider responses that somehow spread the information across the repetition times. We will assume that the excitation sequence induces an embedding for the response map \eqref{eq: full response map}  (here random sequences seem to suffice), and identify additional conditions that enable us to develop a random $k$-space subsampling strategy with an appropriate RIP condition.
Our approach will follow the technique of random sampling as is common in compressed sensing measurement design, along with a pre-conditioning technique that has been used in the Fast Johnson-Lindenstrauss Transform \cite{Ailon-FJLT2009} and in spread spectrum compressed sensing \cite{Puy-2012}. It is also reminiscent of the Rauhut's bounded orthonormal systems \cite{Rauhut2010} and has a similar aim of ensuring that information is sufficiently spread within the measurement domain

The key vectors of interest are those that discriminate between pairs of possible signals within our model, namely the \emph{chords} of $\R_+\B$, which are the vectors of the form $u = X_{i,:}-\tilde{X}_{i,:}$ with $X_{i,:}, \tilde{X}_{i,:} \in \R_+\B$ and $X_{i,:} \neq \tilde{X}_{i,:}$. We will quantify the pre-conditioning requirement for the excitation response through the \emph{flatness} of such vectors which we define as follows.

\begin{Definition}
Let $U$ be a collection of vectors $\{u\}$ in $\C^L$. We denote the \emph{flatness}, $\lambda$, of the these vectors by:
\begin{equation}
 \lambda := \max_{u \in U} \frac{\| u\|_{\infty}}{\|u\|_2}.
\end{equation}
Note that from standard norm inequalities $L^{-1/2} \leq \lambda \leq 1$.
\end{Definition}

We will consider the chords of an excitation response to be sufficiently flat up to a log penalty if $\lambda \sim L^{-1/2} \log^\alpha L$ for $U = \{\R_+\B-\R_+\B\}\backslash\{0\}$.

In constructing our measurement function we also note that the signal model contains no spatial structure, and therefore we should expect to have to uniformly sample $k$-space in order to achieve a sufficient RIP. Note this is in contrast with the variable density sampling strategy proposed by \cite{Ma-MRF2013} which concentrated samples at the centre of $k$-space. It turns out that we can achieve this using a remarkably simple random subsampling pattern based on multi-shot Echo-planar Imaging (EPI) \cite{McKinnon: multi-shot EPI}. 

Let $F\in \C^{N \times N}$ denote the 2D discrete Fourier transform (assuming an image size of $\sqrt{N} \times \sqrt{N}$)
with $F_{i,:}$, $i = 1, \ldots, N$ denoting the $N$ 2D discrete Fourier basis vectors associated with the spatial frequencies $k_x(i),k_y(i) \in \{0, \ldots, \sqrt{N}-1\}$. Without loss of generality we assume that the vectors are ordered such that $k_x(i) = (i-1) \mod \sqrt{N}$, and $k_y(i) = \lfloor (i-1)/ \sqrt{N} \rfloor$.  We can now define a \emph{random Echo-Planar Imaging} measurement operator by $Y_{:,l} = P(\zeta_l)F X_{:,l}$, where $\zeta_l$ is a sequence of independent random variables uniformly drawn from $\{0, \ldots , p-1\}$ and $P(\zeta) \in \R^{M\times N}$ is defined as follows:
\begin{equation}
P_{i,j} = \left\{ \begin{array}{ll}
1& \mbox{if~} j =  i+\sqrt{N}\bigl(\zeta+(p-1)\lfloor (i-1)/\sqrt{N} \rfloor \bigr) \\
0& \mbox{otherwise}.
\end{array}               \right.
\end{equation} 
where for convenience we have assumed that $N$ is exactly divisible by $p$ so that $M=N/p$ is an integer.
%
%
In words, we uniformly subsample $k_y$ by a factor of $p$ with random shifts across time in $k_y$ of the set of $k$-space samples. This is illustrated in figure~\ref{fig: random EPI pattern}.

\begin{figure}[htbp]
\begin{center}
\includegraphics[width=0.7\linewidth]{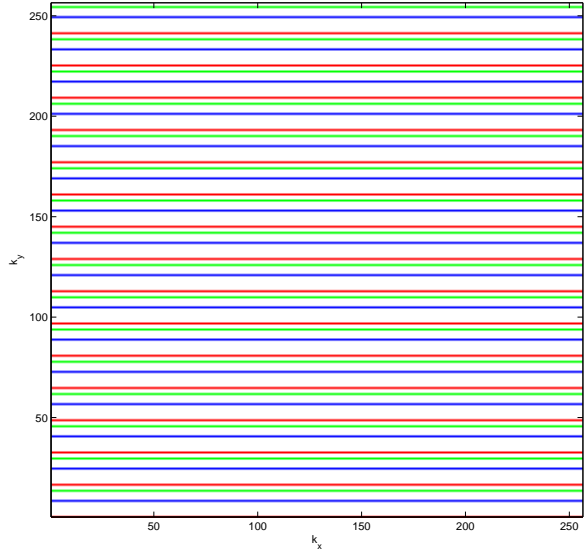}
\caption{The plot shows an instance of random EPI $k$-space sampling for three time frames: red, green and blue respectively. A colored pixel indicate that that $(k_x,k_y)$ frequency is sampled at the associated time frame through the projection operator, $P(\zeta_l)$. In this instance $p=16$.} 
\end{center}
\label{fig: random EPI pattern}
\end{figure}

Random EPI, along with an excitation response with appropriate chord flatness, $\lambda$, is then sufficient to provide us with a measurement operator, $h$, that is a bi-Lipschitz embedding on our signal model. In appendix~\ref{app: RIP proof} we prove the following theorem:

\begin{Theorem}[RIP for random EPI]\label{th: randEPI embedding}
Given an excitation response cone, $\R_+ \B$ of dimension $d_\B$, whose chords have a flatness $\lambda$ and a random EPI operator $h: (\R_+ \B)^N \rightarrow \C^{M \times L}$, then, with probability at least $1-\eta$, $h$ is a restricted isometry on $(\R_+ \B)^N - (\R_+ \B)^N$ with constant $\delta$ as long as:
\begin{equation}
\lambda^{-2} \geq C \delta^{-2} p^2 d_\B \log ( N /\delta \eta )
\end{equation}
for some constant $C$ independent of $p, N, d_\B, \delta$ and $\eta$.
\end{Theorem}

Specifically, if $\lambda = \mathcal{O}( L^{-1/2} \log^\alpha L)$ then we require:
\begin{equation} 
L = \mathcal{O} (\delta^{-2} p^2 d_\B \log( N/\delta \eta)\log^{\alpha}(L))
\end{equation}
excitation pulses. While we might hope to get $L$ of the order of $p d_\B$ it appears that this is not possible, at least for a worst case RIP analysis based on the flatness criterion alone. Indeed, in the experimental section we will provide evidence to suggest that $L\sim p^2$ is indeed the scaling behaviour that we empirically observe.

\begin{Remark}
It might seem surprising that the proposed scheme uses uniform random sampling in $k$-space whereas it is usually advisable to use a variable density sampling strategy for compressed sensing solutions for MRI. Indeed, there is good theoretical justification for variable density sampling patterns \cite{Adcock-2011,Puy2011}. Our theory above is not inconsistent with such results. Variable density sampling is advantageous because the underlying signal model - sparsity in the wavelet domain - is not incoherent with the Fourier basis \cite{Puy2011,Adcock-2011}. However, the Fourier basis is incoherent with a voxel-wise signal model as used above. This is not to say that spatial structure cannot be effectively exploited within a compressed quantitative imaging scheme or that variable density sampling would not then be of benefit. However, as the basic MRF based model does not exploit spatial structure we argue that uniform random sampling is appropriate here. 
\end{Remark}

The challenge of incorporating spatial regularity into the signal model is discussed next.

%
%
%
%
%
%
%
%
%

\subsection{Extending the Bloch response model}
\label{sec: model extension}
Our current compressed sensing model takes no account of additional structure within the parameter maps. This structure could, for example, be the piecewise smoothness of the parameter maps or the magnetization response maps, or an imposed segmentation of the image into different material compositions. In general, it is not clear how such additional regularization can be included in a principled manner, although many heuristic approaches could of course be adopted, as for example in \cite{Doneva2010}. This is because the parameter values are encoded within the samples of the Bloch response manifold, and therefore the spatial regularity would need to be mapped through the Bloch response leading to a non-separable high dimensional nonlinear signal model.

The one exception, which we consider here, is the regularization of the proton density map, or at least a close relative. We note, however, that in this instance the theory relies on the real non-negative proton density model and does not directly extend to the complex case.


Let us define the \emph{pseudo-density}, $\tilde{\rho}$ as the proton density map scaled by the norm of the Bloch response
vector, so that:
\begin{equation}
\tilde{\rho}_i = \rho_i \|B({\bf \theta}_i; \alpha, \TR)\|_2.
\end{equation}
Similarly we can define the normalized Bloch response as:
\begin{equation}
\eta_{i,:} = \tilde{B}({\bf \theta}_i; \alpha, \TR) \triangleq B({\bf \theta}_i; \alpha, \TR)/\|B({\bf \theta}_i; \alpha, \TR)\|_2
\end{equation}
and the normalized Bloch response manifold, $\tilde{\B}$ as:
\begin{equation}
\tilde{\B} = \left\{\eta_{i,:} = \tilde{B}({\bf \theta}_i; \alpha, \TR) \mbox{~ for some~} \theta_i \in \M\right\}
\end{equation}

The pseudo-density will be roughly the same as the density, as long as the Bloch response sequences are all of approximately the same magnitude. The transform to $\{\tilde{\rho},\eta\}$ normalizes the manifold $\tilde{\B}$ so that we can more easily calculate projections onto product signal models of the form $\{\tilde{\rho},\eta\} \in \Sigma \times \tilde{\B}^N$, where $\Sigma$ denotes the set of spatially regularized pseudo-density maps.
To do this we will find the following proposition useful:

\begin{Proposition}
Given an $X \in \C^{N \times L}$, suppose that the projection onto the signal model $\Sigma \times \tilde{\B}^N$ is given by $\hat{\tilde{\rho}} \in \Sigma$ and $\hat{\eta}_{i,:} \in \tilde{\B}$ and results in $\hat{\tilde{\rho}}_i \geq 0$ for all $i$, then:
\begin{equation}\label{eq: eta proj}
\hat{\eta}_{i,:} = \argmax_{\eta_{i,:} \in \tilde{\B}}  z_i
\end{equation}
and
\begin{equation}\label{eq: rho proj}
\hat{\tilde{\rho}} = \argmin_{\tilde{\rho} \in \Sigma} \| \tilde{\rho} - z\|_2^2
\end{equation}
where $z_i = \Real \langle \eta_{i,:},X_{i,:} \rangle$.
\end{Proposition}

\begin{proof}
By definition of the orthogonal projection we have:
\begin{equation}\label{eq: orth proj def} 
\{\hat{\eta},\hat{\tilde{\rho}}\} = \argmin_{\eta,\tilde{\rho}} \sum_i \sum_j |X_{i,j} - \tilde{\rho}_i \eta_{i,j}|^2
\end{equation}
Expanding \eqref{eq: orth proj def}, substituting in $z_i$ and noting that $\|\eta_{i,:}\|_2 = 1$ we have:
\begin{equation} \label{eq:orth proj expansion}
\{\hat{\eta},\hat{\tilde{\rho}}\} = \argmin_{\eta \in {\cal{\B}},\tilde{\rho} \in \Sigma} \sum_i \left(  \tilde{\rho}_i^2 - 2\tilde{\rho}_i z_i  \right).
\end{equation}
By assumption $\hat{\tilde{\rho}}_i$ is non-negative so the expression is minimized with respect to $\eta_{i,:}$ by \eqref{eq: eta proj} independently of ${\tilde{\rho}}_i$. Finally we note that \eqref{eq: rho proj} holds since:
\begin{equation}
\sum_i \left(  \tilde{\rho}_i^2 - 2\tilde{\rho}_i  z_i \right) = \|\tilde{\rho}-z\|_2^2 + \mbox{const.}
\end{equation}
\qquad \end{proof}

One way to impose spatial regularity on $\tilde{\rho}$ is to force it to be sparse in the wavelet domain for some appropriate orthogonal wavelet representation, $W$, such that $c = W \tilde{\rho}$. In this case, the projection \eqref{eq: rho proj} can be written as $\hat{\tilde{\rho}} = W^T \hat{c}$ with:
\begin{equation}\label{eq: hard threshold}
\hat{c} = {\cal H}_k (W z )
\end{equation}
where ${\cal H}_k$ denotes an element-wise hard thresholding \cite{TB-IHT-2008,TB-IHT-2009} that retains only the largest $k$ elements. 

Under the non-negativity assumption the projection operator can be formed by applying \eqref{eq: eta proj} followed by \eqref{eq: hard threshold}. This results in a simple algorithm for incorporating a degree of spatial regularization within the compressed quantitative imaging framework. In the next section we will see, however, that the inclusion of this addition spatial constraint adds little to the performance of the compressed sensing approach, suggesting that the Bloch equation constraint dominates the performance.

\begin{Remark}
The above calculation is only guaranteed to be valid when the resulting pseudo-density map is non-negative. In theory, applying such an operator when we incur negative values of pseudo-density could give a projection that is not optimal. However, in practice we have found that this is not a problem as we always impose non-negativity on both the pseudo-density and the correlations with the Bloch response, $z_i$, in order to ensure that the projection is physically meaningful.  
\end{Remark}

%
%


\section{Experiments}
\label{sec: experiments}
In order to test the efficacy of BLIP for compressed quantitative imaging we performed a set of  simulations using an anatomical brain phantom, segmented into various material compositions. This provided a well defined ground truth and enabled us to demonstrate image sequence recovery and parameter map estimation as a function of the $k$-space subsampling factor and the excitation sequence lengths. 

\subsection{Experimental Set up}
The key ingredients of the experimental set up are described below.
 
\subsubsection*{Anatomical Brain Phantom}

To develop realistic simulations that also provide a solid ground truth we have adapted the anatomical brain phantom of \cite{Collins-1998}, available at the BrainWeb repository \cite{Brainweb}. A $217 \times  181$ slice (slice 40) of the crisp segmented anatomical brain was used and restricted to contain only 6 material components, listed in table~\ref{table: material props}. The phantom was further zero padded to make a $256\times 256$ image to simplify the computations. Since we are using the crisp segmentation the model is somewhat idealized and does not address inaccuracies associated with partial volume effects or many of the other issues with real MRI. However, it serves as a useful test bed to provide a good proof-of-concept for our proposed techniques.

The material properties were chosen to be both representative of the correct tissue type \cite{Hornak-MRI webbook} and challenging: the proton densities were fixed to give little discrimination for individual parameters and were set so that there is not an exact match to the sampling of the Bloch response manifold. 

\begin{table}
\caption{Tissue types used from MNI segmented brain phantom} 
\centering 
\begin{tabular}{|l|c||c|c|c|}
 \hline 
 \rule[-1ex]{0pt}{2.5ex} Tissue & index & proton density & $\Tone$ (ms) & $\Ttwo$ (ms) \\ 
 \hline 
 \hline
 \rule[-1ex]{0pt}{2.5ex} Background & 0 & 0 & - & - \\ 
 \hline 
 \rule[-1ex]{0pt}{2.5ex} CSF & 1 & 100 & 5012 & 512 \\ 
 \hline 
 \rule[-1ex]{0pt}{2.5ex} Grey matter & 2 & 100 & 1545 & 83 \\ 
 \hline 
 \rule[-1ex]{0pt}{2.5ex} White matter & 3 & 80 & 811 & 77 \\ 
 \hline 
 \rule[-1ex]{0pt}{2.5ex} Adipose & 4 & 80 & 530 & 77 \\ 
 \hline 
 \rule[-1ex]{0pt}{2.5ex} Skin/Muscle & $5/6$ & 80 & 1425 & 41 \\ 
 \hline 
 \end{tabular}  
 \label{table: material props} 
\end{table}

The segmented brain is shown, colored by index, in figure~\ref{fig: segmented brain phantom}.

\begin{figure}[htbp]
\begin{center}
\includegraphics[width=\linewidth]{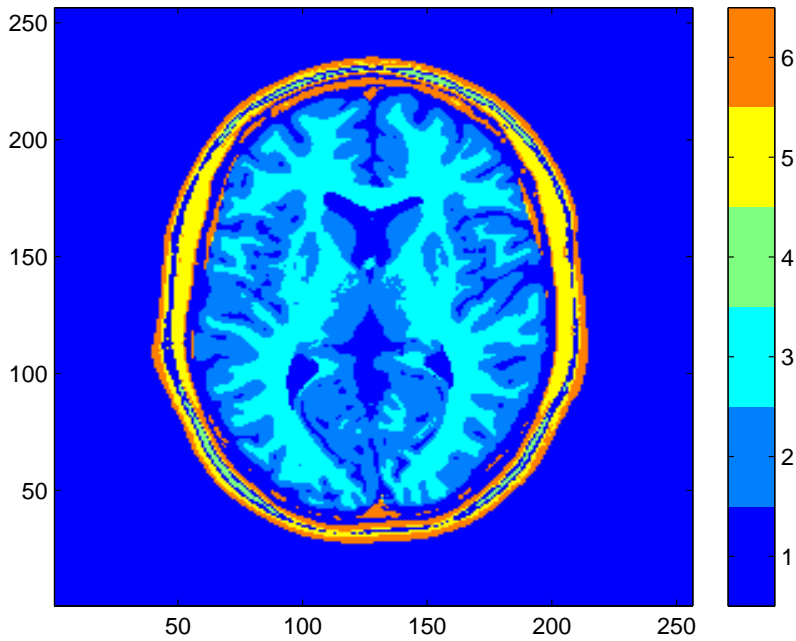}
\caption{The MNI segmented anatomical brain phantom \cite{Collins-1998} colored by index: $0=$Background, $1=$CSF, $2=$Grey Matter, $3=$White Matter, $4=$Fat, $5=$Muscle/Skin, $6=$Skin.\label{fig: segmented brain phantom}} 
\end{center}
\end{figure}

\subsubsection*{Pulse excitation}

For the excitation sequences we use IR-SSFP sequences (exemplar code can be found in the supplementary material of \cite{Ma-MRF2013}) with random flip angles drawn from an independent and identically distributed Gaussian distribution:
\begin{equation}
\alpha_l \sim \N(0,\sigma_\alpha^2)
\end{equation}
with a standard deviation, $\sigma_\alpha = 10$ degrees. The repetition times were uniformly spaced at an interval of $10$ ms. While we also experimented with randomizing repetition times, we did not find that these significantly changed the performance of the techniques. Constant repetition time intervals also mean that we can directly assess the imaging speed in terms of the sequence length, $L$.

The value of $\sigma_\alpha$ was chosen empirically to provide reasonable persistence of excitation for the expected $\Tone$ and $\Ttwo$ responses. Figure~\ref{fig: sequence flatness} (left) shows the magnitude of the response differences for the set of tissue types listed in table~\ref{table: material props}. It can be seen that the difference in the responses does indeed persist over time. Using these differences we can also estimate their flatness. Figure~\ref{fig: sequence flatness} (right) shows how the flatness varies as a function of sequence length. We see that $\lambda^{-2}$ roughly scales proportionally to $L$, as desired, with a slight downward sublinear trend.

\begin{figure}[t]
\begin{center}
\begin{minipage}{0.49\linewidth}
  \includegraphics[width=\linewidth]{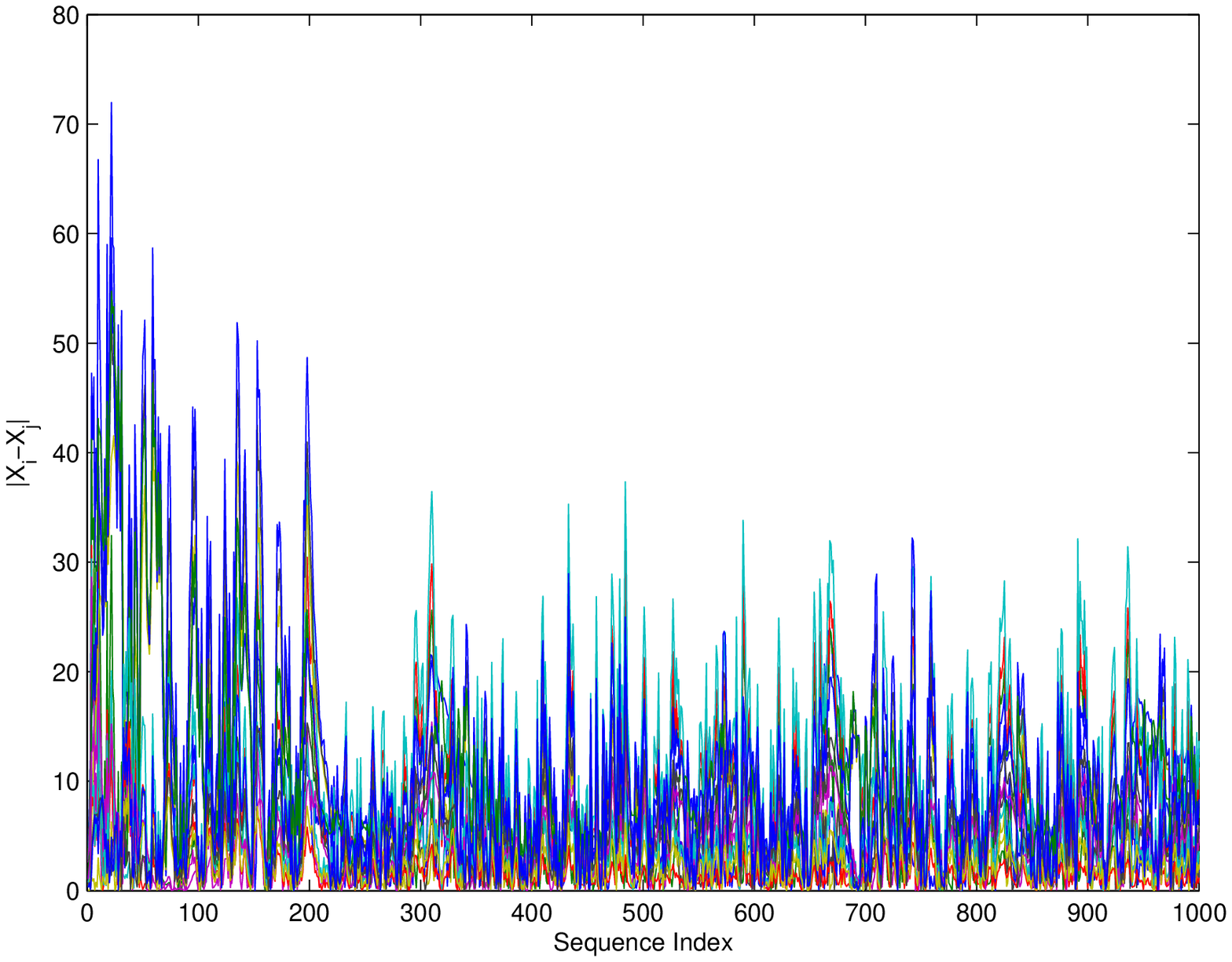}
\end{minipage}
\begin{minipage}{0.49\linewidth}
  \includegraphics[width=\linewidth]{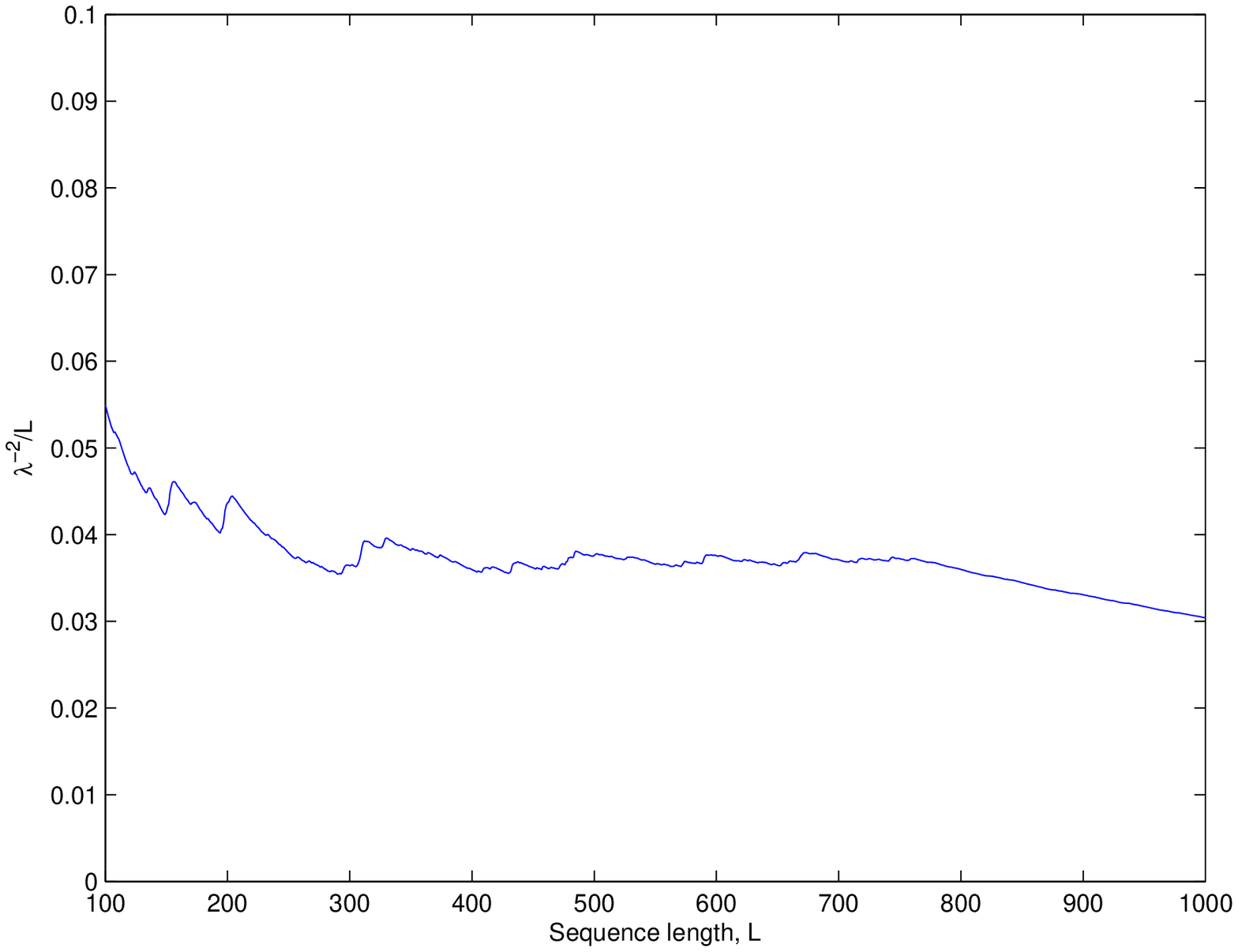}
\end{minipage}
\caption{Left: examples of the response differences for pairs of tissue types given in table~\ref{table: material props} when using IR-SSFP pulse sequence excitation with random flip angles. Right: $\lambda^{-2}/L$ as a function of sequence length for the repsonse differences plotted on the left. From this plot it can be deduced that $\lambda^{-2}$ grows roughly proportionally to $L$.}
\label{fig: sequence flatness}
\end{center}
\end{figure}

\subsubsection*{Discretized Bloch response}
The Bloch response manifold was sampled in a similar manner to \cite{Ma-MRF2013}, however, to simplify things we have only considered variation in $\Tone$ and $\Ttwo$ here, assuming the off resonance frequency is equal to zero. Similar to \cite{Ma-MRF2013}, discrete samples for $\Tone$ were selected to go between 100 and 2000 in increments of 20 and from 2300 to 6000 in increments of 300. $\Ttwo$ was sampled between 20 and 100 in increments of 5, from 110 to 200 in increments of 20 and from 400 to 1000 in increments of 200. This results in a dictionary of size $3379\times L$. This range of $\Tone$ and $\Ttwo$ values clearly spans the anticipated range for the tissue types listed in table~\ref{table: material props}.

\subsubsection*{Subsampling strategy}
For the $k$-space subsampling we use the random EPI sampling scheme detailed in section~\ref{sec: excitation and sampling}. Specifically, we fully sample the $k$-space in the $k_x$ direction while regularly subsampling the $k_y$ direction by a factor of $p$. This deterministic sampling pattern was then cyclically shifted by a random number of $k_y$ lines at each repetition time. In most experiments $p$ is set to $16$ (sampling at $6.25\%$ of Nyquist).

\subsubsection{Reconstruction algorithms}

In the experiments below we compare three distinct algorithms for reconstructing the magnetization image sequences. These are: (1) the original MRF algorithm; (2) BLIP algorithm presented in Algorithm~\ref{alg:CS-MRF}; and (3) BLIP with spatial regularization as detailed in section~\ref{sec: model extension}. For both iterative algorithms we use the adaptive step size strategy set out in section~\ref{sec: adaptive step size} with $\kappa = 0.99$. For the spatial regularization we use a Haar wavelet representation with hard thresholding as detailed in section~\ref{sec: model extension}, retaining only the largest 12000 wavelet coefficients at each iteration.

As the MRF algorithm (with step size equal to $1$) underestimates the value of the image sequence (and also the proton density) we include in the appropriate plots the performance of a rescaled MRF algorithm where the step size is $\mu = N/M$. 

Finally, in some of the plots we also include the performance for an oracle estimator. This oracle is given the fully sampled image sequence data as an input and then projects each voxel sequence onto the discretized Bloch response. In this way we can differentiate between errors associated with the Bloch response discretization and the image sequence reconstruction.

\subsection{Results}

All the experiments were evaluated using a signal-to-error-ratio (SER) in decibels (dBs), calculated as $20 \log_{10} \tfrac{\|x\|_2}{\|x-\hat{x}\|_2}$ for a target signal $x$ with the estimate $\hat{x}$. For $\Tone$ and $\Ttwo$ this corresponds to the measures $\mbox{T}_1\mbox{NR}$ and $\mbox{T}_2\mbox{NR}$ that has been used to gauge the efficiency of relaxation time acquisition schemes \cite{Deoni-DESPOT-2003}. To avoid issues of estimates associated with empty voxels the errors are only calculated over regions with a non-zero proton density value.

In all experiments, unless stated otherwise, the following parameters were used: the undersampling ratio for the operator $h(\cdot)$ was fixed at $1/16$ and for both the iterative algorithms a maximum of $20$ iterations was allowed, though in many cases fewer iterations would have sufficed.

\subsubsection{Performance as a function of excitation sequence length} 

Our first experiment evaluates the performance of the algorithms in terms of the sequence length, which was varied between $10$ and $1000$ pulses. Here we can separately evaluate the performance of the compressed sensing component and the recovery of the parameter maps. 

The compressed sensing recovery performance, evaluated by the SER of the image sequence reconstruction, $X$, is shown in figure~\ref{fig: SER vs sequence length}~(a).

First, note that the strange behaviour of the oracle estimator for small sequence lengths is probably due to the failure of $f(\cdot)$ to achieve a low distortion embedding. This would result in it being easier to approximate voxel sequences with a given  element of the Bloch response approximation. Beyond this the performance reaches a plateau at approximately SER $= 27$ dB which can be considered to be the error associated with the discretization of the Bloch response.

The performance of both BLIP algorithms is roughly equivalent. They both sharply increase in performance at a sequence length of $100$ and then tend to a plateau beyond this with an SER of about $0.5$ dB below that of the oracle estimator. This suggests that we can achieve near perfect compressed sensing reconstruction with a sequence containing as few as 100 pulses. In this simulation there was no significant gain from the additional inclusion of the spatial regularization. 

The performance of MRF is significantly worse. We first highlight that the non-rescaled MRF performance is terrible, however, as noted earlier, this is mainly due to the shrinkage effect of the subsampling operator, $h(\cdot)$. Correcting for this with appropriate rescaling leads to significantly improved estimation. However, we see that the SER increases slowly as a function of sequence length, which is consistent with the argument that the matched filter is averaging over the aliasing rather than cancelling it, as presented in section~\ref{sec: MRF reconstruction}. Furthermore, even for a sequence length of 1000 the SER still only reaches $12$dB.

Subfigures~ \ref{fig: SER vs sequence length} (b), (c) and (d) show the SER for the estimation of the parameter maps, proton density, $\Tone$ and $\Ttwo$ respectively, and reflects the combined performance of inverting both $h(\cdot)$ and $f(\cdot)$. In each case the two iterative algorithms approach the oracle performance for sequence lengths of $L \geq 200$, indicating successful parameter map recovery. Furthermore, the performance for the $\rho$ estimates and $\Ttwo$ estimates do not improve substantially beyond the $L = 200$ value as $L$ is increased reaching a plateau at approximately $16$dB which corresponds to a root mean squared (rms) error of approximately $30$ms. In contrast, the $\Tone$ estimation performance does increase from roughly $20$dB ($213$ms rms error) at $L=200$ to $30$dB ($67$ms rms error) at $L = 1000$. This may be a function of the isometry properties (in the $\Tone$ direction) for the Bloch response embedding, and is possibly related to the longer time constants of $\Tone$. It is an open question as to whether a better excitation sequence can be designed to improve the $\Tone$ estimates for small $L$.

\begin{figure}[t]
\begin{center}
\begin{minipage}{0.49\linewidth}
  \includegraphics[width=\linewidth]{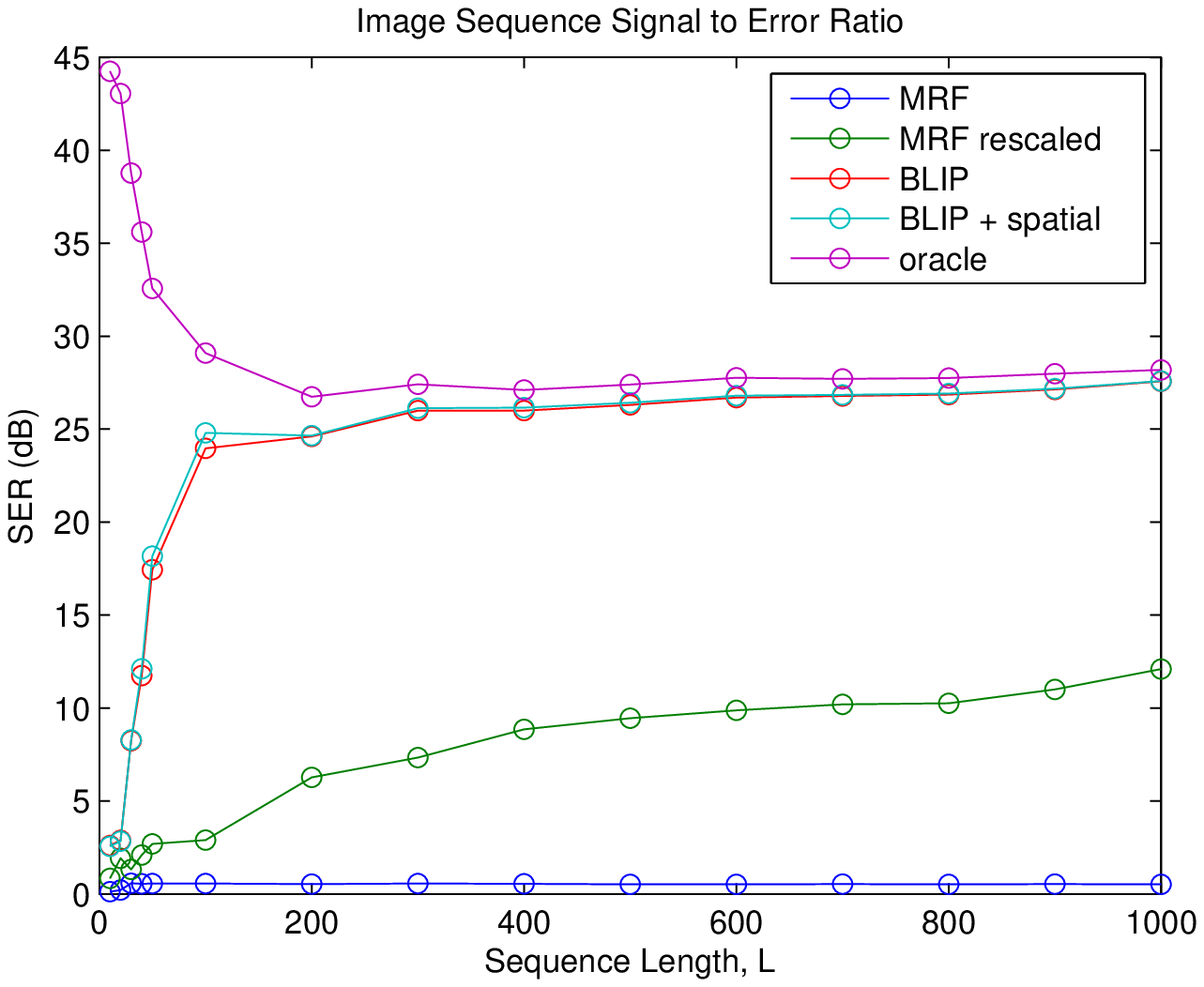}
  \begin{center}
  {(a)}
  \vspace{12pt}
  \end{center}
\end{minipage}
\begin{minipage}{0.49\linewidth}
  \includegraphics[width=\linewidth]{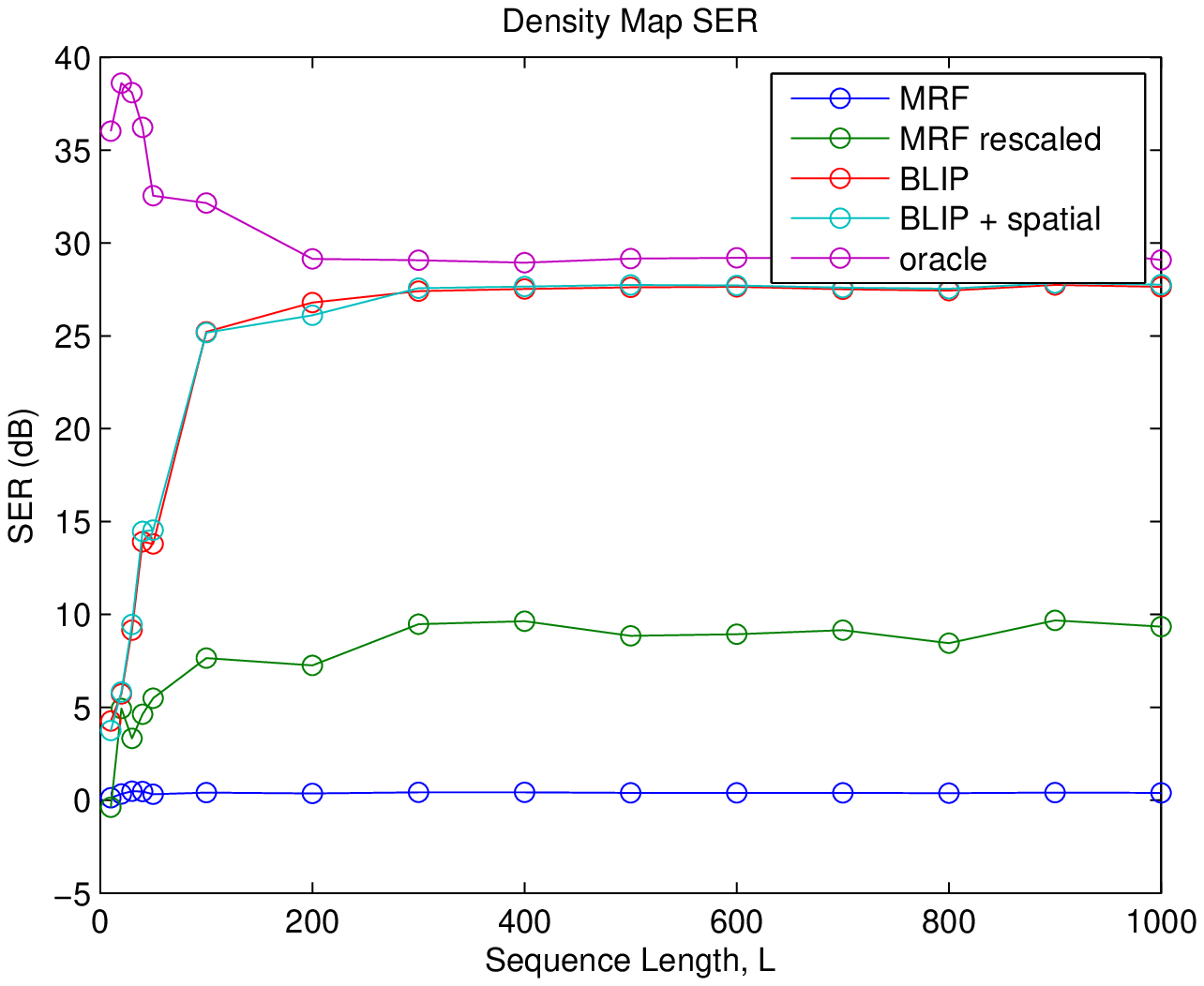}
  \begin{center}
  {(b)}\vspace{12pt}
  \end{center}
\end{minipage}
\begin{minipage}{0.49\linewidth}
  \includegraphics[width=\linewidth]{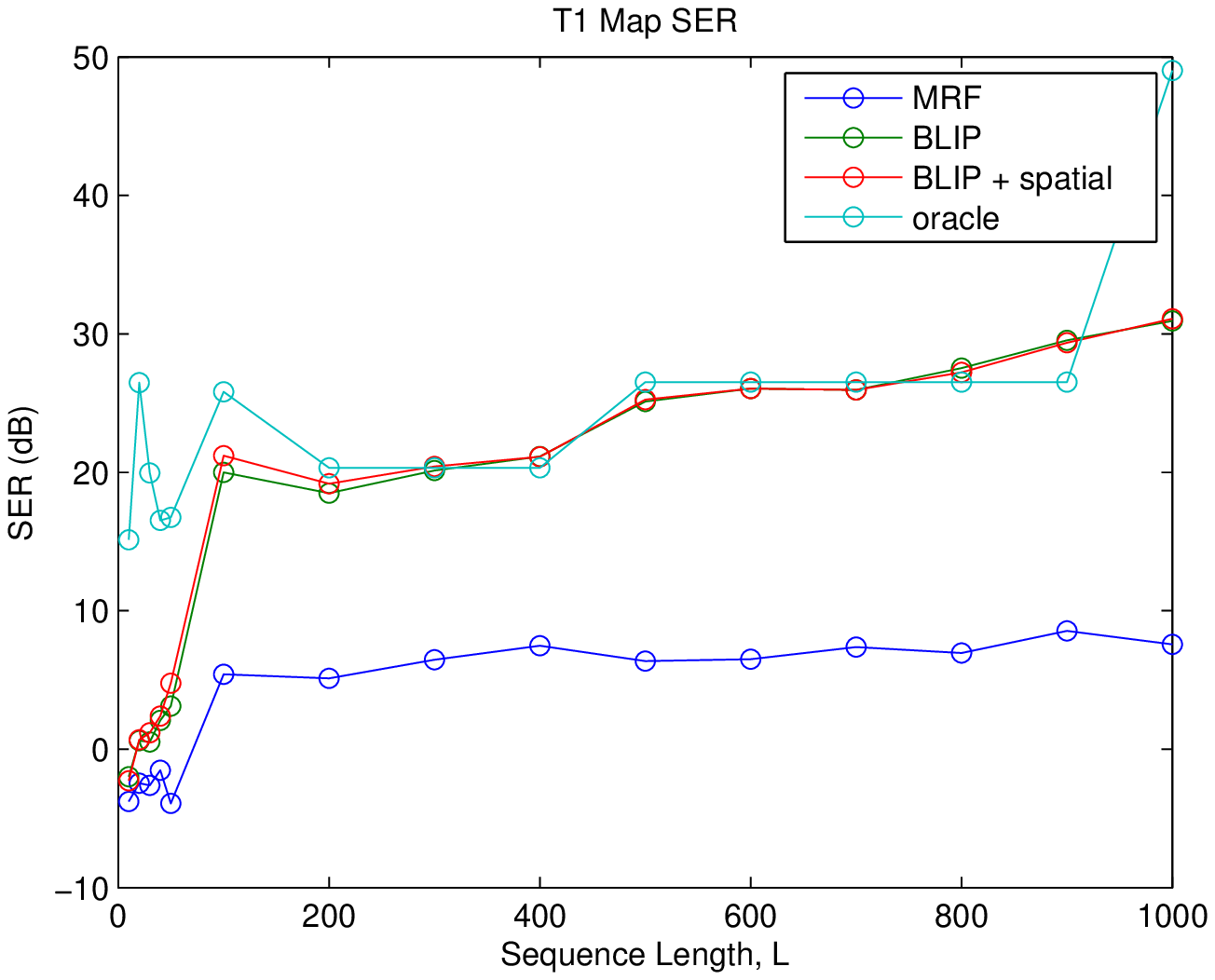}
  \begin{center}
  {(c)} \vspace{12pt}
  \end{center}
\end{minipage}
\begin{minipage}{0.49\linewidth}
  \includegraphics[width=\linewidth]{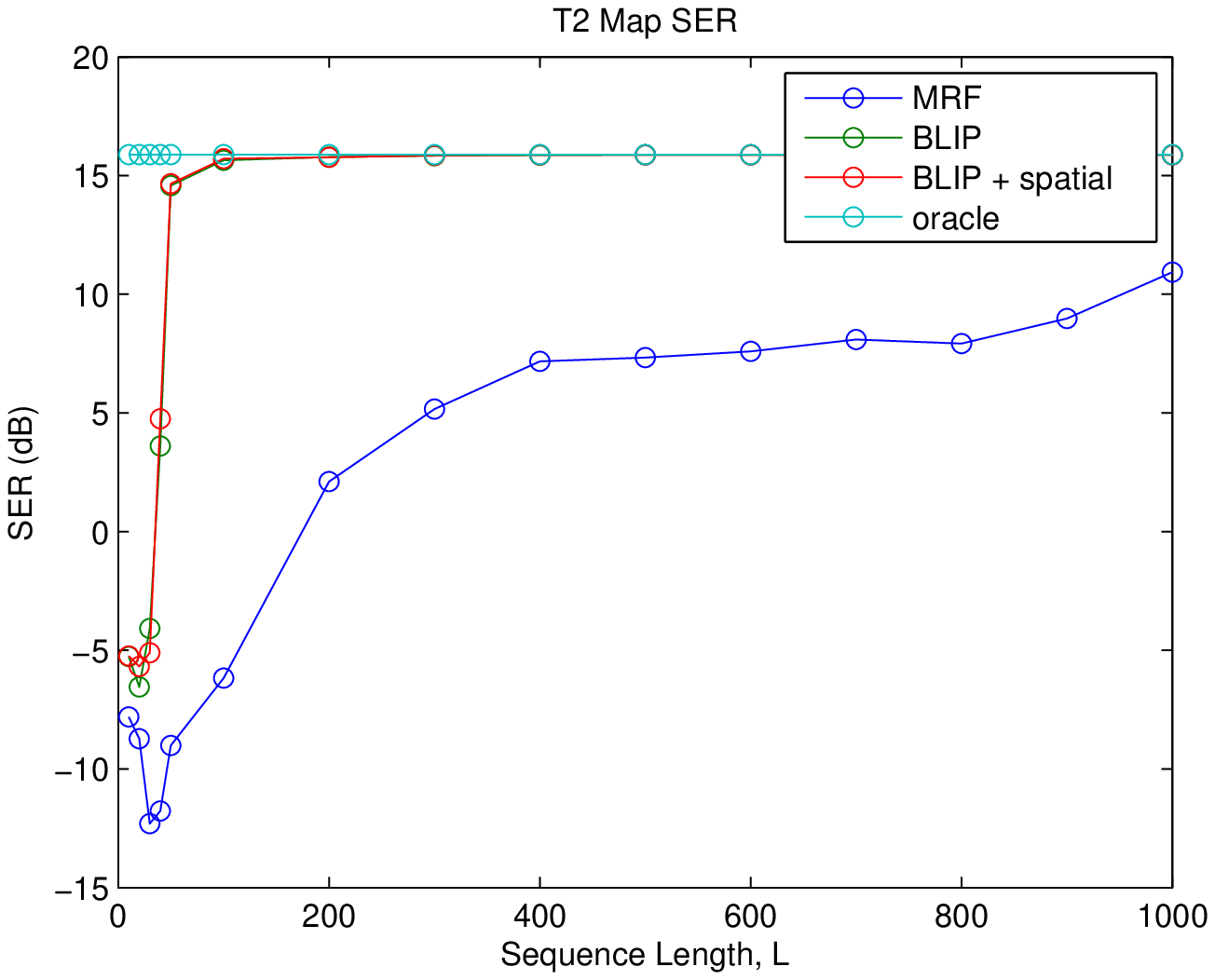}
  \begin{center}
  {(d)} \vspace{12pt}
  \end{center}  
\end{minipage}
\caption{\label{fig: SER vs sequence length}
Reconstruction performance as a function of sequence length. (a) SER for image sequence reconstruction; (b) SER for density map estimation; (c) SER for $\Tone$ map estimation; and (d) SER for $\Ttwo$ map estimation. Results are shown for the following algorithms: MRF, BLIP, BLIP with spatial regularization. Also shown is the performance of an oracle estimator given the full image sequence data. Finally subfigures (a) and (b) also include the performance of a rescaled MRF estimator.}
\end{center}
\end{figure}

\subsubsection{Visual Comparison}

To get a visual indication of the performance of the BLIP approach over the MRF reconstruction at low sequence lengths, 
images of the 3 different parameter estimates for $L = 300$ are given in figures~\ref{fig: visual comparison density}, \ref{fig: visual comparison T1} and \ref{fig: visual comparison T2}. The left hand column shows the ground truth parameter maps while the middle row shows the MRF reconstruction (scaled) and the right hand column shows the BLIP estimates (with spatial regularization). While the main aspects of the parameter maps are visible in the MRF reconstructions, there are still substantial aliasing artefacts. These are most prominent in the $\Tone$ and $\Ttwo$ estimates. In contrast, the BLIP estimates are virtually distortion-free, indicating that good spatial parameter estimates can be obtained with as little as 300 excitation pulses.

\begin{figure}[!ht]
\vspace{-1cm}
\centering
\includegraphics[height=0.99\textheight]{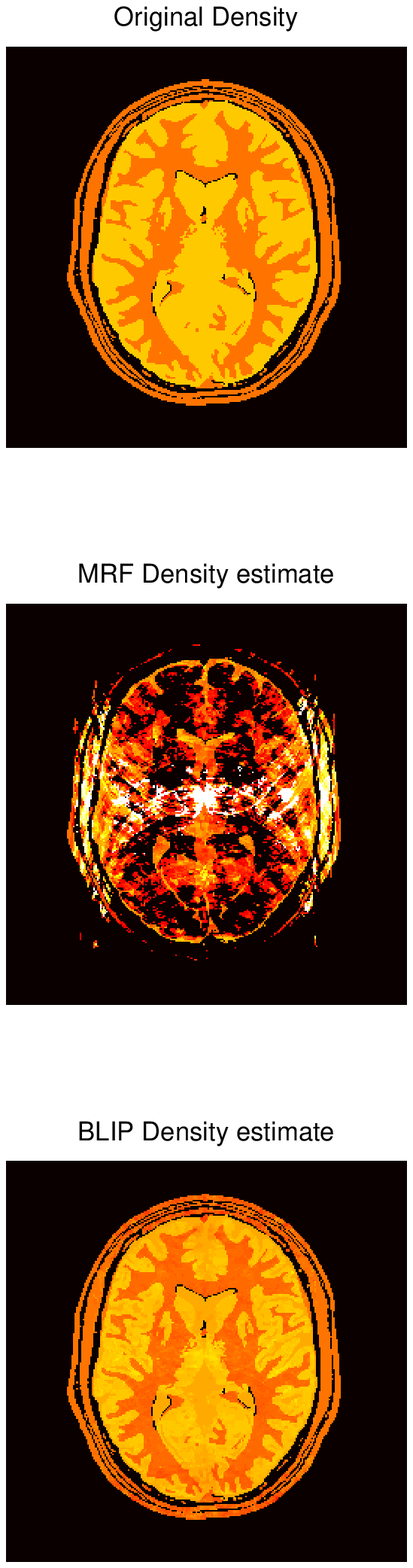}
\vspace{-1cm}
\caption{\label{fig: visual comparison density}
A visual comparison of the density map estimates from a sequence of length $L=300$. The top plot shows the original density map. The middle image is the MRF estimate and the bottom image is the BLIP estimate.}
\end{figure}

\begin{figure}[!ht]
\begin{center}
\vspace{-1cm}
\includegraphics[height=0.99\textheight]{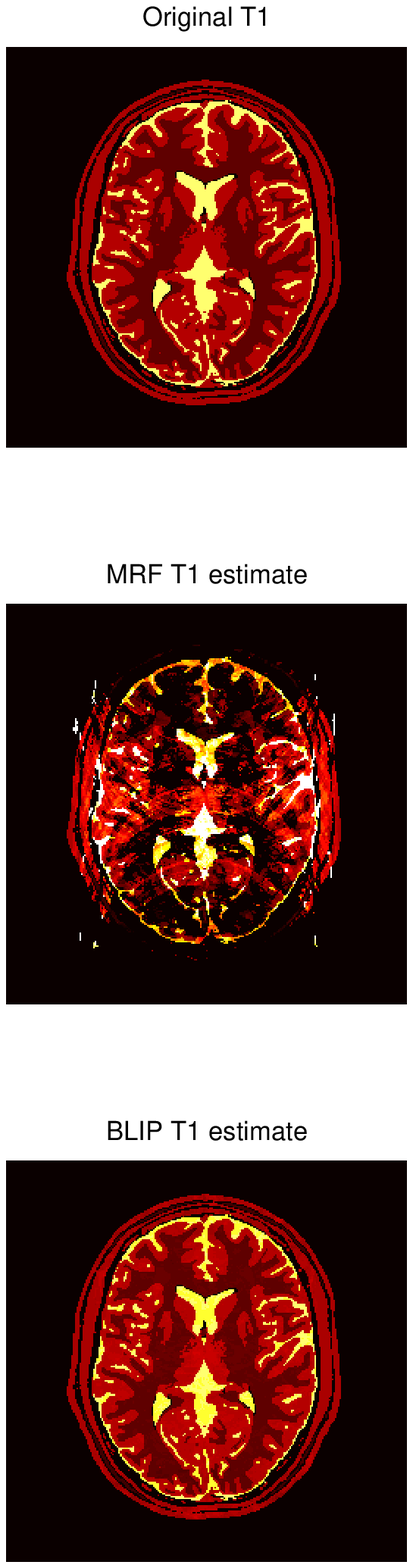}
\vspace{-1cm}
\caption{\label{fig: visual comparison T1}
A visual comparison of the $\Tone$ map estimates from a sequence of length $L=300$. The top plot shows the original $\Tone$ map. The middle image is the MRF estimate and the bottom image is the BLIP estimate.}
\end{center}
\end{figure}

\begin{figure}[!ht]
\begin{center}
\vspace{-1cm}
\includegraphics[height=0.99\textheight]{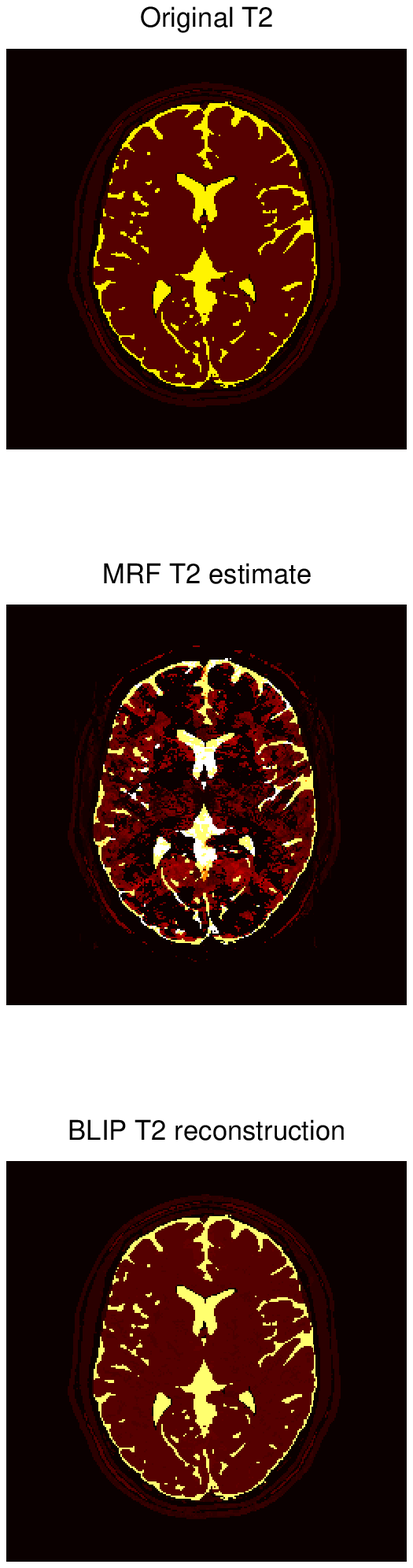}
\vspace{-1cm}
\caption{\label{fig: visual comparison T2}
A visual comparison of the $\Ttwo$ map estimates from a sequence of length $L=300$. The top plot shows the original $\Ttwo$ map. The middle image is the MRF estimate and the bottom image is the BLIP estimate.}
\end{center}
\end{figure}


\subsubsection{Convergence rates for BLIP}

The convergence of the iterative algorithms is shown in figure~\ref{fig: convergence plot} as a function of the relative data consistency error at each iteration $k$, which we define as $\|Y-h(X^{k})\|_2^2/\|Y\|_2^2$. Results for three different sequence lengths, $100$, $200$ and $500$, are shown in the figure. It is clear that in all cases the algorithms converge rapidly and for sequence lengths of $200$ or more have effectively converged within $20$ iterations (note the log scale along the y-axis). Indeed, this is predicted by the compressed sensing theory for IPA: when  the sequence length increases, so that compressed sensing task becomes easier (smaller isometry constant) the rate of convergence also increases. Thus BLIP can be considered to be reasonably computationally efficient. 

\begin{figure}[t]
\begin{center}
\includegraphics[width=0.7\linewidth]{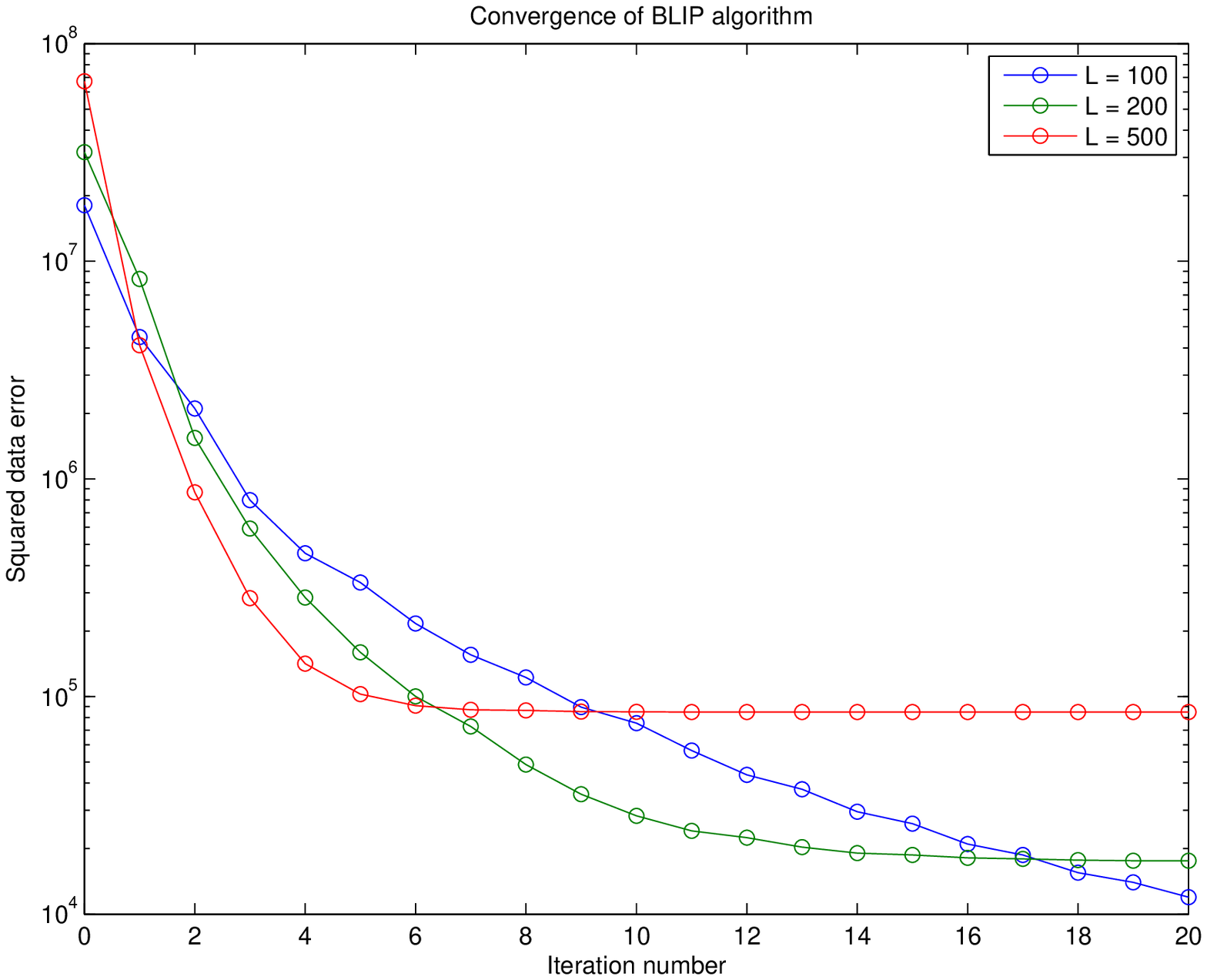}
\caption{\label{fig: convergence plot}
Plots of the data consistency error at each iteration for BLIP using a varying sequence length. The convergence rate increases as the sequence length increases. This is consistent with theory as the increased sequence length is likely to reduce the isometry constant.}
\end{center}
\end{figure}
  
\subsection{Subsampling versus sequence length}

In our next experiment we investigate the dependencies of the undersampling ratio and the sequence length on the reconstruction performance. In this experiment we evaluate the image sequence SER as a function of $L$ and $p$. Recall that the theory presented in section~\ref{sec: excitation and sampling} suggested that this performance might degrade roughly as a function of $p^2/L$. However, as we noted earlier, the analysis in that section is of a `worst case' type and may be highly conservative. Figure~\ref{fig: L vs p2} shows a plot of the image sequence SER as a function of $L/p^2$ for three different subsampling rates: $p = 16$ (green), $p = 32$ (red) and $p = 64$ (blue). From the plot we can see that the rapid growth of the SER that we associate with successful recovery occurs in each case at roughly the same value of $L/p^2$. This seems to suggest that the predicted scaling behaviour for $L$ and $p$ in random EPI to achieve RIP is of the right order. This in turn suggests that to maximize efficiency we should attempt to minimize $p$ (all other design criteria being equal).

\begin{figure}[t]
\begin{center}
\includegraphics[width=0.7\linewidth]{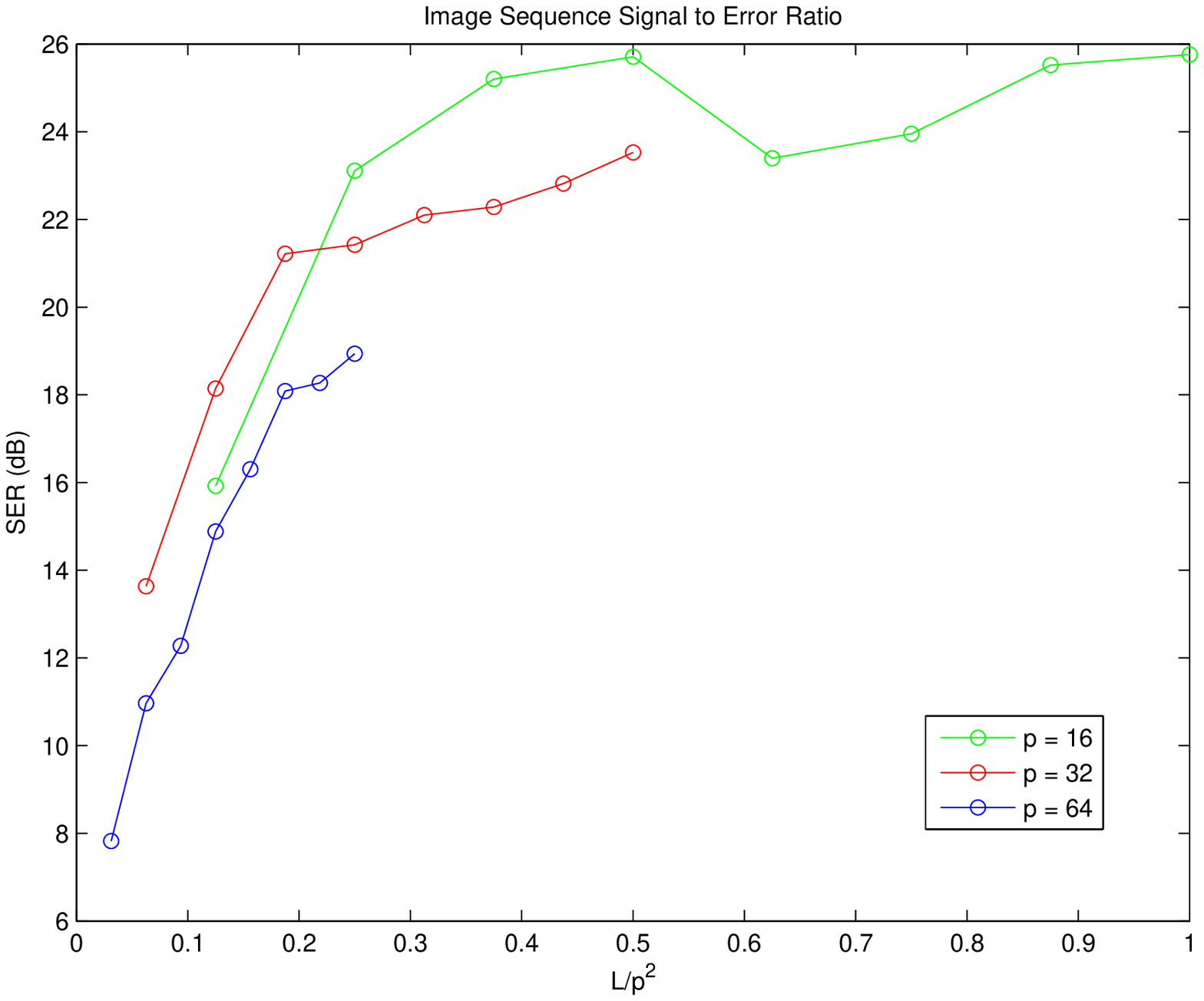}
\caption{\label{fig: L vs p2}
A plot of the Image sequence SER (dB) against $L/p^2$ for three different levels of undersampling: $p = 16$ (green), $p = 32$ (red) and $p = 64$ (blue). The rapid increase in SER appears to occur at roughly the same value of $L/p^2$ in each case suggesting that the RIP result in Theorem~\ref{th: randEPI embedding} is of the right order.}
\end{center}
\end{figure}

%

\subsection{Using a complex density model}
\label{sec: complex density experiment}
The simulations, so far, have used the somewhat idealized model that the density map is real and non-negative. In this experiment we demonstrate that the algorithm works just as well when the density map is allowed to be complex and to absorb sensitivity maps and other phase terms. Here we repeat the first experiment but we modify the density map to have a quadratic phase that is zero at the centre of the image and $pi/4$ at the corners. A plot of the phase is shown on the left hand side in figure~\ref{fig: complex_density_T2_SER}.

We then ran the MRF reconstruction algorithm and BLIP with equations \eqref{eq: bloch MF} and \eqref{eq: bloch MF proton density} replaced by \eqref{eq: bloch complex MF} and \eqref{eq: bloch MF complex proton density} in both algorithms. The resulting performance was very similar to that in the real valued case. For brevity we only show a plot of the the $\Ttwo$ SER in figure~\ref{fig: complex_density_T2_SER}. We see that the parameter estimation behaves  identically to that in the first experiment. Similar behaviour can be observed for the other parameters. Therefore, it seems that there is no significant difference in using the real or complex model for proton density.

\begin{figure}[t]
\begin{center}
\begin{minipage}{0.49\linewidth}
  \includegraphics[width=\linewidth]{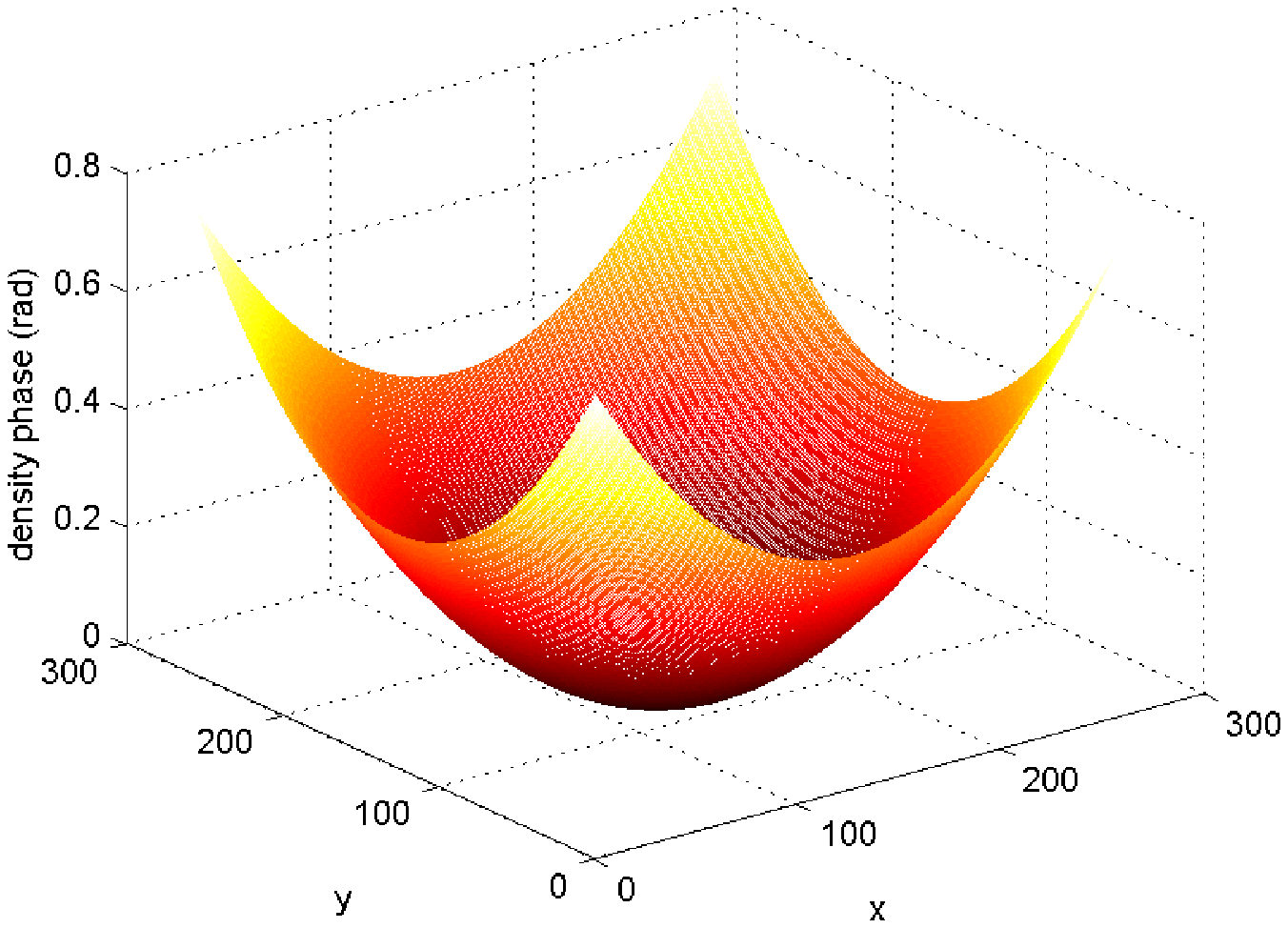}
  \begin{center}
  {(a)}
  \vspace{12pt}
  \end{center}
\end{minipage}
\begin{minipage}{0.49\linewidth}
  \includegraphics[width=\linewidth]{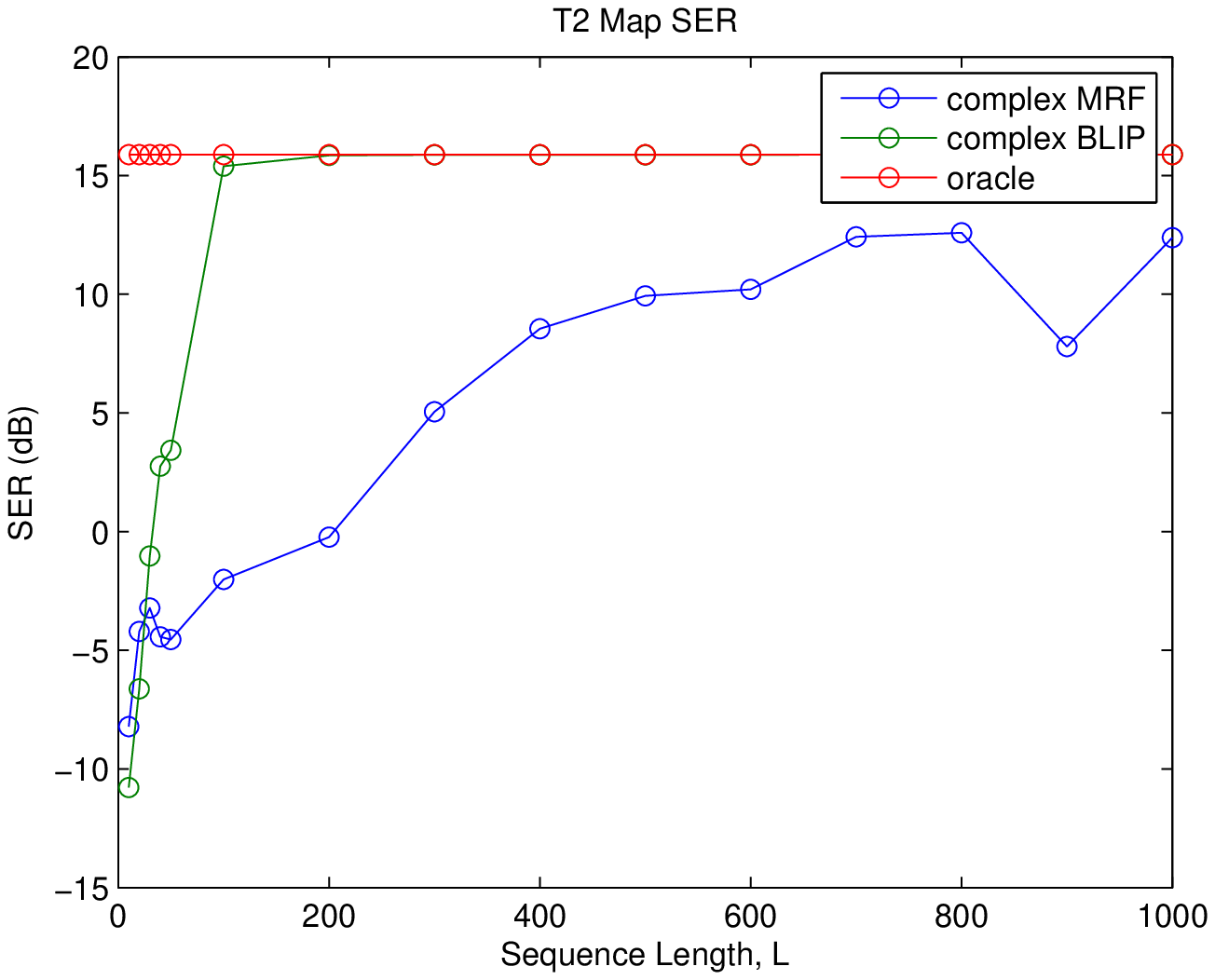}
  \begin{center}
  {(b)}\vspace{12pt}
  \end{center}
\end{minipage}
\caption{
\label{fig: complex_density_T2_SER}
Reconstruction performance for the $\Ttwo$ map using a complex density model. (a) The quadratic phase applied to the density map; (b) SER for $\Ttwo$ map estimation as a function of sequence length. Results are shown for the following algorithms: complex MRF, complex BLIP and the complex oracle estimate.
}
\end{center}
\end{figure}

\subsection{Uniform versus non-uniform sampling}

In \S\ref{sec: excitation and sampling} we asserted that as the Bloch response model does not include any spatial structure it is preferable to take uniformly random samples of $k$-space in order to achieve the RIP rather than use a variable density scheme. In this final experiment we examine the effect of replacing the (uniform) random EPI sampling with a sampling pattern that weights the lower frequencies more, as is common in compressed sensing schemes for MRI \cite{Lustig2008}. Specifically, we choose a non-uniform sampling pattern with an equivalent undersampling ratio, $M/N = 1/16$, that always samples $k_y = 0, 1, 2, \sqrt{N}-3, \sqrt{N}-2 \mbox{~and~} \sqrt{N}-1$ (the centre of $k_y$-space), and then samples the remainder of $k$-space uniformly at random (with the remaining $10$ samples). While we have not tried to optimize this non-uniform sampling strategy we have found that other variable density sampling strategies performed similarly. 

We repeated the first experiment and compared the random EPI sampling to using non-uniform sampling with the sequence length varied between $10$ and $300$. Again we focus on the $\Ttwo$ reconstruction, although similar behaviour was observed for the density and $\Tone$ estimation (not shown). The $\Ttwo$ results are plotted in figure~\ref{fig: uniform_vs_nonuniform}. It is clear from the figure that BLIP does not perform well with the non-uniform sampling of $k$-space, and it never achieves the near oracle performance that we observe with the random EPI sampling strategy. Indeed, we observed no non-uniform sampling strategy to do this. Other simulations (not shown) have indicated that uniform i.i.d. undersampling in $k_y$ also performs well, although we have yet to prove this has the RIP. 

Interestingly, the MRF reconstruction does benefit from the non-uniform sampling, however the reconstruction quality is still very poor. We believe that this can be explained by the fact that in both cases the MRF reconstructions exhibit significant aliasing. However, in the non-uniform case the aliasing is concentrated more in the high frequencies where the signal has less energy and therefore introduces less distortion. 

%
%


\begin{figure}[t]
\begin{center}
\includegraphics[width=0.7\linewidth]{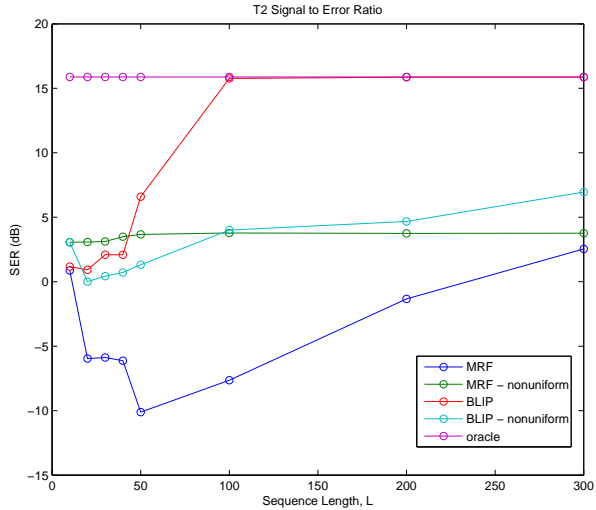}
\caption{\label{fig: uniform_vs_nonuniform}
A plot of the $\Ttwo$ estimate SER (dB) against $L$ for reconstruction algorithms MRF and BLIP using uniform random (EPI) sampling with $p = 16$ and a non-uniform random sampling with an equivalent undersampling ratio $M/N = 16$. Only in the case of BLIP with uniform random sampling does the $\Ttwo$ estimate performance approach that of the oracle estimator.}
\end{center}
\end{figure}

\section{Conclusions and open questions}

We have presented a principled mathematical framework for compressed quantitative MRI based around the recently proposed technique of Magnetic Resonance Fingerprinting \cite{Ma-MRF2013}. The sensing process can be considered in two separate stages. First, the embedding of the parameter information into the magnetization response sequences through the mapping $f(\cdot)$. Second, the compressive imaging of the induced magnetization image sequence. The key elements of our approach have been: the characterization of the signal model through the Bloch response manifold; the identification of a provably good image sequence reconstruction algorithm based on iterative projection; an excitation response condition based on a newly introduced measure of \emph{flatness} to quantify the persistence of the excitation; and a random EPI $k$-space sampling scheme that can be shown to have the necessary RIP condition when the excitation is suitably flat. 


The simulations presented in \S\ref{sec: experiments} show that the proposed technique is capable of achieving good parameter map reconstruction with very short pulse sequences. 
The next step will be to make a thorough comparison on an MRI scanner with MRF and other existing quantitative MRI techniques such as \cite{Deoni-DESPOT-2003}.
 
While the current work is specifically targeted at a compressed sensing framework for MRF, we believe that many elements of it should be more broadly applicable. Specifically, the RIP condition for randomized EPI may well have applications in other MR imaging strategies and the characterization of  excitation response in terms of flatness could prove a useful tool for the analysis of other compressed sensing schemes involving some form of active sensing. 

Finally, the use of parametric physical models (through appropriate discretisation) could be applicable to many areas of compressed sensing beyond MRI. The experience we have gained here suggests that such models can be more powerful than traditional spatial image models, such as wavelet sparsity, that are often found in compressive imaging.

\subsection{Open Questions}
In setting out this compressed sensing framework a number of questions have arisen that we feel should be addressed. We conclude by briefly describing these below.

\subsubsection*{Excitation sequences}

What are the key requirements for the excitation sequences? We have introduced the flatness condition, however, we have so far not exploited randomness in the excitation. This raises the question: does the excitation sequence need to be random? Although randomness seems a natural way to obtain flat responses, it is not clear that it is necessary or even preferable. Random excitations may also be able to provide less stringent sampling conditions in order to provide the RIP. Furthermore, whether deterministic or random, how should we optimize the excitation sequences in order to maximise the performance of the parameter map estimation? This seems to be very much a system identification problem.

\subsubsection*{Improved signal models}

A key question for the Bloch response model is: how densely do we need to sample $\M$? This
  will depend on the response mapping $f$, the undersampling operator $h$ and the performance of the recovery algorithm. It would be interesting to try to quantify these errors using the existing union of subspace compressed sensing theory \cite{BD-UoS-09,TB-2011}.
  
  A second question is: how should we best include additional modelling information? It is clearly desirable to include spatial regularization. However, we have seen in \S\ref{sec: experiments} that the inclusion of our limited spatial regularization within the signal model did not significantly improve performance. On the other hand, this only regularized the density map, whereas, ideally we would like to impose spatial regularity on each of the parameter maps. Unfortunately, a naive construction of such a model would lead to a complex non-separable representation that we cannot easily project onto. Alternatively, we might try to impose block spatial regularity on the image sequence on top of the Bloch response model. This form of spatial regularization was used in \cite{Doneva2010} and appears to have only provided modest performance improvements.
Therefore the question is how to best combine these models to maximize reconstruction performance and can we back this up theoretically?
 
The current signal model is also somewhat idealised.
We have treated the proton density values, $\rho_i$, as nonnegative, following the physics. However, in MRI it is more common to treat $\rho_i$ as a complex value, absorbing various phase factors into the quantity. While our framework easily extends to the complex case as highlighted in \S\ref{sec: Bloch projection} and evaluated in \S\ref{sec: complex density experiment}, it would be interesting to see whether there was a more principled way to deal with such additional phase factors.

Another idealization that is made both here and in the original MRF is that the read out time is assumed negligible with respect to the relaxation times. Depending on the level of undersampling this may not be true. This might introduce significant artefacts. If so, can we modify the signal model to account for this? 

Finally, our model does not account for partial volume effects. These were briefly touched on in the supplementary material of \cite{Ma-MRF2013}, where it was proposed to model individual voxels as a composition of different material components. Such a model is reminiscent of the spatial abundance maps used in hyperspectral imaging. In such a case we are in the realms of compressive source separation \cite{Golbabaee-2013}. Can we formulate a compressive MRF problem that accounts for partial volume effects in a similar manner?
  
\subsubsection*{Subsampling $k$-space}

We have identified certain conditions that guarantee the RIP for random EPI sampling. This allows us to trade off the $k$-space subsampling factor $p=N/M$ with the length of the excitation sequence, $L$. Unfortunately the trade off scales as $L \sim p^2$. It is not clear whether similar guarantees could be achieved from a deterministic sampling sequence or whether this is indeed optimal. It would be more desirable to have a proportional trade off $L \sim p$. Is such a scaling possible? If so, what is the appropriate combination of excitation sequence and sampling strategy? 

Finally, if we can successfully incorporate spatial structure into our signal model, as suggested above, it is very likely that a variable density sampling would be preferable. If so, can we leverage existing theory for variable density sampling \cite{Adcock-2011,Puy2011} to develop principled designs for variable density sampling for compressive MRF?

\appendix
\section{Dynamics of balanced SSFP sequences}
\label{app: bloch dynamics}
Balanced SSFP sequences are popular in MRI and were the basis of the excitation sequences used in MRF \cite{Ma-MRF2013}, although the term `steady state' is somewhat of a misnomer as this refers to the steady state conditions arrived at following periodic excitation with constant $\alpha$ and $\TR$ \cite{Scheffler2003}.

In fact, here we are explicitly interested in the transient dynamics of a non-periodic excitation sequence. This is in contrast with traditional SSFP sequences where transient oscillations are seen as undesirable as they can introduce imaging artefacts \cite{Hargreaves2001}. In this work, as in \cite{Ma-MRF2013}, we will regard the transient behaviour as essential in enabling us to distinguish between different quantitative behaviour.

The transient response can be formally described in terms of a $3$-dimensional linear discrete time dynamical system that we summarize below, see \cite{Hargreaves2001,Ganter,Scheffler2003} for further details. To keep things simple we will assume there is no phase increment between pulses and also that the $l$th echo time, $\TE_l$, is half the $l$th repetition time $\TR_l$.

Following \cite{Hargreaves2001}, let $\m_{l} = (m_{l}^{x},m_{l}^{y},m_{l}^{z})^T \in \R^3$ represent the 3-dimensional magnetization vector for a voxel at the $l$th excitation pulse. In Inversion Recovery SSFP sequences the equilibrium magnetization, $\m_{\mbox{eq}} = [0,0,1]^T$, is initially inverted so that $\m_0 = [0,0,-1]^T$. Then the magnetization after the $l$th RF-excitation is given by the following linear discrete time dynamical system:
\begin{equation}
\label{eq: magnetization response}
\m_{l+1} = R_x(\alpha_l) R_z(\phi_{l}) E_l  \m_{l} + R_x(\alpha_l)(\Id -E_l) \m_{\mbox{eq}}
\end{equation}
where $R_u(\phi)$ denotes a rotation about the $u \in \{x,y,z\}$ axis by an angle $\phi$, $\phi_l = 2\pi \delta f \TR_l$ is the off-resonance phase associated with local field variations and chemical shift effects \cite{Ganter} and $E_l$ is the diagonal matrix characterizing the relaxation process:
\begin{equation}
E_l \defeq 
\begin{pmatrix}
e^{-\TR_l/\Ttwo} &~&~\\
~& e^{-\TR_l/\Ttwo}& ~\\
~&~&e^{-\TR_l/\Tone}
\end{pmatrix}
\end{equation}
where the $\Tone$ relaxation time controls the rate of relaxation along the $z$-axis, while the $\Ttwo$ relaxation time controls the relaxation onto the $z$-axis. 

Finally let $\hat{\m}_{l}$ denote the magnetization at the echo time, $\TE_l$. Then this is given by \cite{Hargreaves2001}:
\begin{equation}\label{eq: TE response}
\hat{\m}_{l} = R_z(\phi_{l}/2) E_l^{1/2}  \m_{l} + (\Id -E_l^{1/2}) \m_{\mbox{eq}},
\end{equation}
with the readout coil measuring $\hat{m}_{l}^{x}+ j \hat{m}_{l}^{y}$. Thus the magnetization dynamics in response to a sequence of RF pulses with flip angles, $\alpha_l$, and repetition times, $\TR_l$, is given by \eqref{eq: magnetization response} and \eqref{eq: TE response} which apart from the input parameters is solely a function of the tissue parameters $\Tone$, $\Ttwo$, and the off-resonance frequency, $\delta f$.

\section{Proof of Theorem \ref{th: randEPI embedding}}
\label{app: RIP proof}
We first introduce the key lemmas that form the main ingredients of the proof. Our approach will
follow the standard route of concentration of measure, $\epsilon$-net and union bound. To this end we will need the following well known Chernoff bound \cite{Dubhashi2009}:

\begin{Lemma}\label{th: chernoff}
Let $X = X_1+X_2+\ldots +X_n$, $0 \leq X_i \leq 1$ with $\mu = \Expect (X)$. Then
\begin{equation}
\Prob(|X-\mu| > \epsilon \mu) \leq 2 \exp \left( -\frac{\epsilon^2 \mu}{3} \right)
\end{equation}
\end{Lemma}

The next lemma establishes a near isometry for a single aliased voxel sequence.

\begin{Lemma}\label{th: single voxel chord isometry}
Let $z \in \C^L$ be a random vector given by:
\begin{equation}
z_i = \frac{1}{p}\sum_k U_{k,i} e^{-j 2 \pi \zeta_i k / p} 
\end{equation}
where $\zeta_i$ are independent random variables drawn uniformly from $\{0, \ldots, p-1\}$ and $U \in \C^{p \times L}$ is a matrix whose rows have flatness $\lambda$. Then, with probability at least $1-2 e^{-\epsilon^2 /(3 p \lambda^2)}$, $z$ satisfies:
\begin{equation}
(1-\epsilon) \|U\|_F^2 \leq p^2 \|z\|_2^2 \leq (1+\epsilon) \|U\|_F^2
\end{equation}
\end{Lemma}

\begin{proof}
We first show that $\Expect \|z\|_2^2 = \frac{1}{p^2}\|U\|_F^2$ and then derive the necessary tail bounds.

Let $W_{a,k} = \frac{1}{\sqrt{p}} e^{-j 2 \pi a k/p}$, $a,k = 0, \ldots, p-1$, denote the unitary Discrete Fourier transform in $\C^p$. We can then write
\begin{align}
\Expect \| z\|_2^2 &= \sum_{a=0}^{p-1} \frac{1}{p} \left( \sum_{i=1}^L \frac{1}{p} | W_{a,:} U_{:,i}|^2 \right)\\
& = \frac{1}{p^2} \sum_i \sum_a | W_{a,:} U_{:,i}|^2\\
&= \frac{1}{p^2} \sum_i \|U_{:,i}\|_2^2\\
&= \frac{1}{p^2} \|U\|_F^2,
\end{align}

Now note that $\|z\|_2^2$ is the sum of $L$ independent random variables, $\|z\|_2^2 = \sum_i \xi_i$ with $\xi_i = \frac{1}{p} |W_{\zeta_i,:} U_{:,i}|^2$. Furthermore the $\xi_i$ satisfy:
\begin{equation}
\begin{split}
0 \leq \xi_i & \leq \frac{1}{p} \|U_{:,i}\|_2^2\\
& \leq \frac{1}{p} \sum_k \max_i |U_{k,i}|^2\\
& \leq \frac{1}{p} \sum_k \lambda^{2} \|U_{k,:}\|_2^2\\
& = \frac{\lambda^2}{p} \|U\|_F^2
\end{split}
\end{equation}

We can therefore apply the Chernoff bound from Lemma~\ref{th: chernoff} to $\sum \xi_i$ rescaled by $\frac{\lambda^2}{p} \|U\|_F^2$ to give:
\begin{equation}
\Prob(|\|z\|_2^2 - \frac{1}{p^2}\|U\|_F^2| > \epsilon \frac{1}{p^2} \|U\|_F^2) \leq 2 \exp \left( - \frac{\epsilon^2}{3 p \lambda^2}\right)
\end{equation}
Rearranging this expression completes the proof.
\qquad \end{proof}

Next we extend Lemma~\ref{th: single voxel chord isometry} to a near isometry for groups of aliased voxels under the action of $h$. Since $h$ is an ortho-projector, $\|h (X)\|_2^2 = \|h^H h(X) \|_2^2$ and so we can equivalently consider the isometry properties of $h^Hh$. 

Let us denote $Z = h^H(h(X))$ such that $Z_{:,l} = F^H P(\zeta_l)^T Y_{:,l}$. Recall that $h$ is a partially sampled 2D discrete Fourier transform that is fully sampled in the $k_x$ direction and periodically subsampled by a factor of $p = N/M$ in the $k_y$ direction.  Therefore each $Z_{i,l}$ is the sum of $p$ aliases taken from $X_{:,l}$:
\begin{equation}
Z_{i,l} = \frac{1}{p} \sum_{k=0}^{p-1} X_{\tau_i(k),l}~ e^{-j 2 \pi \zeta_l k / p}
\end{equation}
where $\tau_i(k)$ gives the index of the $k$th alias for the $i$th voxel (with $\tau_i(0) = i$). We can therefore partition the set $\{1,\ldots,N\}$ into $M$ disjoint index sets $\Lambda_1, \ldots, \Lambda_M$ with each set associated with $p$ aliases, such that $h^Hh$ is separable over $\{\Lambda_i\}$ and $Z_{\Lambda_i,:} = [h^Hh]_{\Lambda_i} X_{\Lambda_i,:}$. Since each $Z_{\Lambda_i,:}$ contains $p$ copies of the same combination of aliases (up to a phase shift) we can conclude that:
\begin{equation}
\|Z_{\Lambda_i,:}\|_F^2 = p \|Z_{k,:}\|_2^2, ~\forall k \in \Lambda_i
\end{equation}

Applying Lemma~\ref{th: single voxel chord isometry} then gives us:
\begin{Lemma}\label{th: aliased voxel isometry}
Let $Z_{\Lambda_i,:} = [h^Hh]_{\Lambda_i} X_{\Lambda_i,:}$ for some $X_{\Lambda_i,:} \in \C^{p\times L}$ whose rows have a flatness $\lambda$ where $[h^Hh]_{\Lambda_i}$ is defined above. Then with probability at least $1-2 e^{-\epsilon^2 /(3 p \lambda^2)}$ we have 
\begin{equation}\label{eq: single vector isometry}
(1-\epsilon) \|X_{\Lambda_i,:}\|_F^2 \leq p \|Z_{\Lambda_i,:}\|_F^2 \leq (1+\epsilon) \|X_{\Lambda_i,:}\|_F^2
\end{equation}
\end{Lemma}

The final ingredient guarantees a near isometry for low dimensional subsets of the unit sphere (for a more sophisticated but slightly different result in this direction see \cite{Clarkson-manifolds-2008})
\begin{Lemma}\label{th: boxdim isometry}
Let $S \subset \mathbb{S}^{n-1}$ have box counting dimension $d$ such that for any $\epsilon > 0$ there exists an $\epsilon$-cover of $S$ of size $C_S \epsilon^{-d}$. Let $P:\C^n\rightarrow \C^k$ be a random projection such that for any $\delta>0$ and a fixed $x \in S$,
\begin{equation} \label{eq: single point isometry}
1-\delta \leq {\frac{n}{k}}\|Px\|_2^2 \leq 1+\delta
\end{equation}
holds with probability at least $1-c_0 e^{-c_1 \delta^2}$. Then $P$ satisfies \eqref{eq: single point isometry} for all $x \in S$ with probability at least $1-\eta$ as long as:
\begin{equation}
c_1 \geq 72 \delta^{-2} \left( d \log (36n/\delta k) + \log C_S c_0/\eta \right)
\end{equation}
\end{Lemma}

\begin{proof}
Consider an $\epsilon$-cover $S_\epsilon$ of $S$ with $\epsilon = \delta'/(2\sqrt{n/k})$ and suppose that $P$ satisfies 
\begin{equation}
1-\delta'/2 \leq {\frac{n}{k}}\|Px\|_2^2 \leq 1+\delta'/2
\end{equation}
for all $x \in S_\epsilon$ with a constant $0<\delta'<1$. Then there exists a $u \in S_\epsilon$ such that:
\begin{align}
\sqrt{\frac{n}{k}} \| Px\|_2 &\leq \sqrt{\frac{n}{k}}\| Pu\|_2 +\sqrt{\frac{n}{k}}\| P(x-u)\|_2\\
& \leq 1+ \delta'/2 + \sqrt{\frac{n}{k}} \epsilon \label{eq: squared RIP implies a non-squared RIP}\\
&= 1+ \delta'
\end{align}
where in \eqref{eq: squared RIP implies a non-squared RIP} we have used the fact that $(1+\delta'/2)^2 > (1+\delta'/2)$.

We can similarly show that $\sqrt{\frac{n}{k}}\|Px\|_2 \geq 1-\delta'$. Then finally noting that the "non-squared" RIP implies the squared RIP in \eqref{eq: single point isometry} with $\delta = 3\delta'$ gives us the required isometry.

It only remains to bound the probability of failure. Let $p_f$ be the probability that $P$ fails to satisfy \eqref{eq: single point isometry} on $S$. By the union bound:
\begin{align}
pf &\leq |S_\epsilon| c_0 e^{-c_1 (\delta'/2)^2}\\
&\leq C_S c_0 \left(\frac{\delta'}{2\sqrt{n/k}}\right)^{-d} e^{-c_1 (\delta'/2)^2}
\end{align} 
Therefore it is sufficient to choose $\eta$ so that:
\begin{equation}
\frac{\eta}{C_S c_0} \geq  \left(\frac{\delta}{6\sqrt{n/k}}\right)^{-d} e^{-c_1 (\delta/6)^2}
\end{equation}
Re-arranging the above gives:
\begin{equation}
c_1 \geq 72 \delta^{-2} \left( d \log (36n/\delta k) + \log C_S c_0/\eta \right)
\end{equation}
as required. \qquad \end{proof}

We are now ready to prove the main theorem.

\begin{proof}
[Proof of Theorem~\ref{th: randEPI embedding}]

First, note that $\R_+\B \subset \R \B$ which is an infinite union of subspace model, as is its $p$-product, $(\R\B)^p$ associated with a group of aliased voxels, $\Lambda_i$. To guarantee that $h_{\Lambda_i}$ possesses the necessary RIP on $(\R\B)^p-(\R\B)^p$ it is sufficient to consider the RIP on the normalized difference set $S$ given by:
\begin{equation}
S = \{x \in ((\R\B)^p-(\R\B)^p),\|x\|_2 = 1\},
\end{equation}
due to the linearity of $h$. 

By construction we have $\dim(S) = 2p d_\B-1$ and we can therefore apply Lemma~\ref{th: boxdim isometry} to $S$ together with Lemma~\ref{th: aliased voxel isometry}. This guarantees for all $X_{\Lambda_i,:} \in (\R\B)^p-(\R\B)^p$
that $h$ satisfied \eqref{eq: single vector isometry} with probability at least $1-\eta$ as long as:
\begin{equation}
\lambda^{-2} \geq (3p) \times 72 \delta^{-2} \left( (2p d_\B-1) \log (36p/\delta) + \log C_S c_0/\eta \right)
\end{equation}
To ensure this holds for all aliased voxel groups $\Lambda_i$, $i = 1,\ldots, M$ we can again apply the union bound and replace $\eta$ by $M \eta$. Noting that $p, \delta^{-1}, \eta^{-1} > 1$ we can collect together the constants and simplify to finally give:
\begin{equation}
\lambda^{-2} \geq C \delta^{-2}  p^2 d_\B  \log (N/ \delta \eta )
\end{equation}
for some constant $C$ independent of $p, N, d_\B, \delta$ and $\eta$ which gives the required conditions of the theorem.
\qquad \end{proof}




\bibliographystyle{siam}


%







\end{document}